\documentclass[11pt,reqno]{amsart}
\usepackage[left=2cm,right=2cm,top=2cm,bottom=2cm]{geometry}
\usepackage{mathtools}
\mathtoolsset{centercolon}
\usepackage{amsmath,amsthm,amssymb,bm,bbm,dsfont}
\usepackage[mathscr]{eucal}
\usepackage{upgreek}
\usepackage{setspace}
\usepackage{graphicx}
\usepackage{verbatim}
\usepackage{float}
\usepackage{placeins}
\usepackage{array}
\usepackage{booktabs}

\usepackage{threeparttable}
\usepackage[update,prepend]{epstopdf}
\usepackage{multirow}
\usepackage{amsfonts,amssymb,dsfont}
\usepackage[abs]{overpic}

\usepackage[usenames,dvipsnames]{xcolor}
\usepackage{colortbl}
\usepackage[hidelinks,bookmarks=false]{hyperref}
\hypersetup{
	bookmarks=false,         
	unicode=false,          
	pdftoolbar=true,        
	pdfmenubar=true,        
	pdffitwindow=false,     
	pdfstartview={FitH},    
	pdftitle={My title},    
	pdfauthor={Author},     
	pdfsubject={Subject},   
	pdfcreator={Creator},   
	pdfproducer={Producer}, 
	pdfkeywords={keyword1} {key2} {key3}, 
	pdfnewwindow=true,      
	colorlinks=true,        
	linkcolor=Red,          
	citecolor=ForestGreen,  
	filecolor=Magenta,      
	urlcolor=BlueViolet,    
}
\usepackage{caption, subcaption}
\usepackage{enumitem}

\makeatletter
\newcommand{\opnorm}{\@ifstar\@opnorms\@opnorm}
\newcommand{\@opnorms}[1]{%
	\left|\mkern-1.5mu\left|\mkern-1.5mu\left|
	#1
	\right|\mkern-1.5mu\right|\mkern-1.5mu\right|
}
\newcommand{\@opnorm}[2][]{%
	\mathopen{#1|\mkern-1.5mu#1|\mkern-1.5mu#1|}
	#2
	\mathclose{#1|\mkern-1.5mu#1|\mkern-1.5mu#1|}
}
\makeatother

\newcommand\xqed[1]{%
	\leavevmode\unskip\penalty9999 \hbox{}\nobreak\hfill
	\quad\hbox{#1}}
\newcommand\Endremark{\xqed{$\Diamond$}}

\makeatletter
\ifx\@NODS\undefined%

\let\Diamond=\diamondsuit
\else%
\fi
\makeatother

\makeatletter
\ifx\@NODS\undefined%

\let\mathbb=\mathds
\else%
\fi
\makeatother

\DeclareMathOperator{\Tr}{Tr}
\DeclareMathOperator{\tr}{\overline{tr}}

\DeclareMathOperator{\Ent}{Ent}

\newcommand{\ket}[1]{| #1 \rangle}

\newcommand{\be}{{\mathbf e}}

\def\0{{\mathbf{0}}}
\def\1{{\mathbf{1}}}
\def\2{{\mathbf{2}}}
\def\3{{\mathbf{3}}}
\def\4{{\mathbf{4}}}
\def\5{{\mathbf{5}}}
\def\6{{\mathbf{6}}}

\def\7{{\mathbf{7}}}
\def\8{{\mathbf{8}}}
\def\9{{\mathbf{9}}}

\def\be{\begin{equation}}
\def\ee{\end{equation}}
\def\bea{\begin{eqnarray}}
\def\eea{\end{eqnarray}}

\theoremstyle{plain}
\newtheorem{theo}{Theorem} 
\newtheorem{prop}[theo]{Proposition} 
\newtheorem{lemm}[theo]{Lemma} 
\newtheorem{coro}[theo]{Corollary} 

\theoremstyle{definition}

\theoremstyle{remark}
\newtheorem{remark}{Remark}[section]

\begin{document}
	
\let\origmaketitle\maketitle
\def\maketitle{
	\begingroup
	\def\uppercasenonmath##1{} 
	\let\MakeUppercase\relax 
	\origmaketitle
	\endgroup
}
	

\title{\bfseries \Large{Matrix Poincar\'e, $\Phi$-Sobolev Inequalities, and Quantum Ensembles}}

\author{ {Hao-Chung Cheng$^{1,2,3}$ and Min-Hsiu Hsieh$^1$}}
\address{\small  	
	$^{1}$Centre for Quantum Software and Information (UTS:Q$\ket{\text{SI}}$), \\
	Faculty of Engineering and Information Technology, University of Technology Sydney, Australia\\
	$^{2}$Graduate Institute Communication Engineering, National Taiwan University, Taiwan (R.O.C.)\\
	$^{3}$Department of Applied Mathematics and Theoretical Physics, University of Cambridge, United Kingdom
}
\email{\href{mailto:HaoChung.Ch@gmail.com}{HaoChung.Ch@gmail.com}}
\email{\href{mailto:Min-Hsiu.Hsieh@uts.edu.au}{Min-Hsiu.Hsieh@uts.edu.au}}



\begin{abstract}
Sobolev-type inequalities have been extensively studied in the frameworks of real-valued functions and non-commutative $\mathbb{L}_p$ spaces, and have proven useful in bounding the time evolution of classical/quantum Markov processes, among many other applications.  In this paper, we consider yet another fundamental setting {\textemdash} matrix-valued functions {\textemdash} and prove new Sobolev-type inequalities for them. Our technical contributions are two-fold: (i) we establish a series of matrix Poincar\'e inequalities for separably convex functions and general functions with Gaussian unitary ensembles inputs; and (ii) we derive $\Phi$-Sobolev inequalities for matrix-valued  functions defined on Boolean hypercubes and for those with Gaussian distributions. 
Our results recover the corresponding classical inequalities (i.e.~real-valued functions) when the matrix has one dimension. Finally, as an application of our technical outcomes, we derive the upper bounds for a fundamental entropic quantity {\textemdash} the Holevo quantity {\textemdash} in quantum information science since classical-quantum channels are a special instance of matrix-valued functions. This is obtained through the equivalence between the constants in the strong data processing inequality and the $\Phi$-Sobolev inequality. 
\end{abstract}

\maketitle


	
\section{Introduction} \label{sec:intro}	

\emph{Sobolev-type inequalities} generally refer to upper bounding the uncertainty of real-valued functions $f:\mathcal{X}\to \mathbb{R}$ by its \emph{Dirichlet energy}, while the \emph{hypercontractive inequality} indicates upper bounding the lower moment of $f$ by its higher moment. These functional inequalities were originally motivated by the study of the semiboundedness of certain quantum field Hamiltonians \cite{Ed64}. Substantial progress was independently made by several authors in the 1970's \cite{SH72, Gro72, Bon70,Bec75} (see also the review literature \cite{DGS92,Gro78, Gro93,Bak94, GZ03, Gro14,BGL13}). In particular,  Federbush proved that the logarithmic Sobolev (log-Sobolev) inequality implies the hypercontractivity \cite{Fed69}.  Later, Gross' groundbreaking work \cite{Gro75}  showed  the equivalence between the log-Sobolev inequality and the hypercontractive inequality. 



\noindent  Since then, log-Sobolev inequalities have become a rich topic that has influenced numerous areas of research \cite{ABC+00, Rag13b, Rag14, Gro14}. For example, hypercontractivity was applied to
bounding the local influence of Boolean functions through a global variation in theoretical computer science \cite{KKL88, BRW08, MO10}. In physics, the Poincar\'e and log-Sobolev inequalities provide useful tools for characterising the time evolution of dynamical systems. The convergence rate of a Markov process to its equilibrium can thus be determined by the constants of the Sobolev inequalities \cite{Zeg90, DS96, GZ03, BGL13}.

Motivated by the studies of quantum fermion fields \cite{Gro72, CL93, Bia97},  recent research interest has switched to these functional inequalities in the non-commutative setting. In the late 20th century, Olkiewicz and Zegarlinski completed a crucial step by generalising Gross' log-Sobolev inequalities \cite{Gro75} in non-commutative $\mathbb{L}_p$ space, and proved their equivalence with hypercontractivity under some regular conditions \cite{OZ99}.  New results and progress has thus emerged in  this line of research \cite{BZ00,Car04,CS07,AZ14,CM15} (see also \cite{Zeg16} and the references therein). Specifically, noncommutative log-Sobolev inequalities enable the analysis of the dynamics of von Neumann algebras, including their most interesting case {\textemdash} the time evolution of a single quantum state through a Markov process \cite{Lin76}.  Kastoryano and Temme derived mixing time bounds in some quantum channel semigroups \cite{KT13}. The exponential decay phenomena of various entropic quantities were later discovered \cite{TKR+10, Mon12, KT13,CKMT15, MFW15a, MFW15b,MF16}. 
Moreover, the equivalence of log-Sobolev inequalities and hypercontractivity for completely bounded norms of quantum channel semigroups was proven under certain assumptions \cite{Kin14, SK15}.

Although commutative and noncommutative Sobolev-type inequalities have succeeded in characterising the dynamics of classical and quantum Markov systems,  these functional inequalities cannot be applied to another fundamental scenario {\textemdash} the time evolution of an ensemble of matrices. This scenario includes classical-quantum channels as a special case; hence, it plays a substantial role in quantum information processing  and communication tasks since investigating the dynamics of a quantum ensemble would enable us to capture the information capability of a communication or computation process \cite{CHT15}. In this paper, we pioneer the study of Sobolev-type inequalities in the framework of ensemble matrices, and establish the corresponding matrix Poincar\'e, $\Phi$-Sobolev, and log-Sobolev inequalities. Our result can be viewed as a direct generalisation of real-valued functions $f$ in classical inequalities with matrix-valued functions $\bm{f}:\mathcal{X}\to\mathbb{C}^{d\times d}$. Consequently, we can recover the classical case when $d=1$.
Furthermore, we give a partial answer to the equivalence between hypercontractivity and the log-Sobolev inequality; we show that the matrix Bonami-Gross-Beckner inequality by Ben-Aroya {et al.} \cite{BRW08} implies the proposed matrix $\Phi$-Sobolev inequality. As an application, we show how the Holevo quantity of  a quantum ensemble changes as it evolves according to a classical Markov kernel on its classical labels and a post-selection rule. 

The proposed generalisations could have potential impacts in other research areas as well. For example, in many branches of science and engineering, observed data are more efficiently represented as matrices, such as Dyson's model of matrix-valued Brownian motion \cite{Dys62b}, collaborative filtering \cite{RS05}, group Lasso \cite{YL06}, and multi-class learning \cite{AEP08}. Thus, system performance can be determined through analysis if a random assumption is placed on the matrices \cite{FGT07}.  A recent review of modern random matrix concentration theory surveyed the most successful methods from these topics and provided interesting examples that these techniques can illuminate \cite{BLM13,Tro15, CH16a}.

\subsection{Our Results}
The contributions of this paper are two-fold. 

\begin{enumerate}
	
\item We derive various functional inequalities, including matrix Poincar\'e inequalities, matrix Sobolev inequalities, and matrix log-Sobolev inequalities. 
\begin{itemize}
\item We prove a  Poincar\'{e} inequality for matrix-valued functions in Theorem \ref{theo:Poincare} (see also Corollary~\ref{coro:Poincare_multi}), which generalises the classical Poincar\'{e} inequality \cite{Bob96, Led97}:
		\begin{equation}\label{eq_classical_poincare}
		\textnormal{Var}(f(X)) \leq \mathbb{E}\left[\left\|\nabla f(X) \right\|^2 \right],
		\end{equation}
		where   $X\triangleq (X_1,\ldots,X_n)$ denotes an independent random vector with each $X_i$ taking values in the interval $[0,1]$. Our proof, paralleling its classical counterpart, relies on the matrix-valued Efron-Stein inequality (Theorem~\ref{theo:Efron}). Both Theorem \ref{theo:Poincare} and Corollary~\ref{coro:Poincare_multi} recover the classical Poincar\'{e} inequality (\ref{eq_classical_poincare}) when the matrix dimension $d=1$. We also derive Poincar\'e inequalities for matrix-valued functions with additional assumptions of pairwise commutation (Corollary \ref{coro:Poincare_multi}). Finally, we derive a matrix Gaussian Poincar\'e inequality for Gaussian unitary ensembles (Theorem~\ref{theo:Gaussian}).
		
\item  We prove a $\mathrm{\Phi}$-Sobolev inequality  for matrix-valued functions defined on the Boolean hypercube in Theorem~\ref{theo:Sobolev}, from which we can extend to a $\mathrm{\Phi}$-Sobolev inequality for Gaussian distributions (Theorem \ref{theo:Sobolev_Gau}).  Our $\mathrm{\Phi}$-Sobolev inequality is defective (see Remark \ref{remark:LSI} for a discussion on \emph{tight} and \emph{defective} $\Phi$-Sobolev inequalities), but again  it recovers the classical $\mathrm{\Phi}$-Sobolev inequality when $d=1$. Our proof builds upon a powerful matrix Bonami-Gross-Beckner inequality \cite{BRW08}, from which the hypercontractivity inequality for matrix-valued functions on Boolean hypercubes can be obtained.
		The matrix log-Sobolev inequalities in Corollaries \ref{coro:Log} and \ref{coro:Log_Gau} follow immediately from Theorems \ref{theo:Sobolev} and \ref{theo:Sobolev_Gau}.
		
	\end{itemize}
	
	\item Finally, we connect matrix $\Phi$-entropies to quantum information theory. When $\Phi(x)=x\log x$ and the random matrix $\bm{\rho}_X\equiv\{p(x),\bm{\rho}_x\}_{x\in\mathcal{X}}$, where each $\bm{\rho}_x\succeq 0$ and $\Tr\bm{\rho}_x=1$ is a quantum ensemble, $H_{\Phi}(\bm{\rho}_X)$ is equal to the Holevo quantity $\chi(\{p,\bm{\rho}\})$ (up to a constant dimensional factor for purely technical purposes). If the ensemble $\bm{\rho}_Y\equiv\{q(y),\bm{\sigma}_y\}_{y\in\mathcal{Y}}$ is obtained by evolving $\bm{\rho}_X$ with a Markov kernel $K(y|x)$:
	\begin{eqnarray*}
		q(y) = \sum_{x} p(x) K(y|x) ,\quad \text{and} \quad
		\bm{\sigma}_y = \sum_x \bm{\rho}_x K^*(x|y) 
	\end{eqnarray*} 
	where $K^*(x|y)$ is the backward channel of $K$, then the Holevo quantity of $\chi(\{p,\bm{\rho}\})$ is bounded from above by a constant $c$ times the average Holevo quantity of the ensembles that come from post-selecting the original $\{p,\bm{\rho}\}$  by the postselection rule $K^*$. Moreover, the constant $c$ is related to the ratio of the Holevo quantities $\chi(\{p,\bm{\rho}\})$ and $\chi(\{q,\bm{\sigma}\})$ (see Proposition~\ref{prop:functional_SDPI_cqstate}). This bears a stronger form of the classical strong data processing inequality  \cite{Rag13,Rag14}.

\end{enumerate}

\subsection{Prior work}

\begin{enumerate}
	\item Very few matrix concentration results have been established for general matrix-valued functions. To the best of our knowledge, the only gem in this area is a family of polynomial Efron-Stein inequalities for random matrices \cite{PMT14}, where the theory of exchangeable pairs is used in the proof. In a previous work, we proved the subadditivity property of operator-valued $\Phi$-entropies, and derived an operator Efron-Stein inequality (see also Theorem~\ref{theo:Efron}).
	
	\item We would also like to point out that the matrix $\mathrm{\Phi}$-entropies defined in the paper are different from the entropy functions in the non-commutative $\mathbb{L}_p$ spaces discussed in \cite{CL93, Bia97, BRW08, MO10, Mon12, KT13, Kin14, CKMT15}.  
	Hence, our functional inequalities in Section~\ref{APP} are incomparable with those in the non-commutative $\mathbb{L}_p$ spaces.

\end{enumerate}

This paper is organised as follows: Section~\ref{Preliminaries} reviews the matrix algebra necessary for the remaining paper. We derive Sobolev-type inequalities for matrix-valued functions in Section \ref{APP}.  We connect matrix $\Phi$-entropies to quantum information theory and derive an upper bound for the Holevo quantity in Section~\ref{sec:FI}. Conclusions are given in Section \ref{sec:conclusion}. Appendix \ref{app_lemmas} collects useful lemmas. 

\section{Preliminaries} \label{Preliminaries}

In this section we present the background information necessary for this paper. Basic notations are introduced in Section \ref{notataion}. We then review operator algebra with a focus on \emph{Fr\'echet derivatives} and the convexity properties of matrix-valued functions in Section \ref{calculus} and \ref{Functions}, respectively.

\subsection{Notation} \label{notataion}

Given a separable Hilbert space $\mathcal{H}$, denote by $\mathbb{M}$ the Banach space of all linear operators on $\mathcal{H}$. The set $\mathbb{M}^\text{sa}$ refers to the subspace of the self-adjoint operators in $\mathbb{M}$. We denote by $\mathbb{M}^+$ (resp.~$\mathbb{M}^{++}$) the set  of positive semi-definite (resp.~positive-definite) operators in $\mathbb{M}^\text{sa}$.
If the dimension $d$ of a Hilbert space $\mathcal{H}$ needs special attention, it is highlighted in subscripts, e.g.,~$\mathbb{M}_d$ denotes the Banach space of $d\times d$ complex matrices.
The trace function $\Tr:\mathbb{M}\rightarrow \mathbb{C}$ is defined as
$\Tr\left[ \bm{M} \right] \triangleq  \sum_{k} e_k^* \bm{M} e_k \quad \text{for } \bm{M} \in\mathbb{M}$,
where $(e_k)_k$ is any orthonormal basis of $\mathcal{H}$.
If we focus on finite dimensional Hilbert spaces, then the trace function acting on $\bm{M}$ is equal to the sum of its eigenvalues.
In this paper, we introduce a normalised trace function $\tr$ for every matrix $\bm{M}\in\mathbb{M}_d$ as
$\tr \left[\bm{M} \right] \triangleq  \frac1d \Tr\left[ \bm{M} \right]$.
For $p\in [1,\infty)$, the Schatten $p$-norm of an operator $\bm{M}\in\mathbb{M}^\text{sa}$ is denoted as
$\|\bm{M}\| _{p} \triangleq  \left( \Tr |\bm{M}|^p \right)^{1/p}$,
where $|\bm{M}|\triangleq  \sqrt{\bm{M}^2}$.

We define $\mathbb{S}^n$ as the set of all mutually commuting $n$-tuple self-adjoint operators; namely, if $\vec{\bm{X}}=(\bm{X}_1,\ldots,\bm{X}_n)\in\mathbb{S}^n$, then $[\bm{X}_i,\bm{X}_j]=\mathbf{0}$ for $i \neq j\in \{1,\ldots,n\}$.
We denote by $\mathbb{S}_d^n$ the set of mutually commuting $n$-tuple $d\times d$ Hermitian matrices.
For $\bm{A},\bm{B}\in\mathbb{M}^\text{sa}$, $\bm{A}\succeq \bm{B}$ means that $\bm{A}-\bm{B}$ is positive semi-definite. Similarly, $\bm{A} \succ \bm{B}$ means $\bm{A} - \bm{B}$ is positive-definite.
We denote by $\bm{A}\preceq_\text{U} \bm{B}$, $\bm{A},\bm{B} \in \mathbb{M}^\text{sa}_d$ when $\sigma_i(\bm{A}) \leq \sigma_i(\bm{B})$, $i = 1,\ldots, d$, where $\sigma(\bm{A}) \triangleq (\sigma_1(\bm{A}),\ldots, \sigma_d(\bm{A}))$ are the eigenvalues of $\bm{A}$ arranged in decreasing order.

Throughout this paper, italic capital letters (e.g.~$\bm{X}$) are used to denote operators and the underlined italic capital letters (e.g.,~$\vec{\bm{X}}$) are used to denote a vector of, say $n$, operators.

\subsection{Matrix Calculus} \label{calculus}
In this section, we only provide information relating to the matrix calculus. For a general treatment of this topic, please refer to \cite[Section 2.1]{Ber77}, \cite[Chapter 17]{KA82}, \cite[Section X.4]{Bha97}, \cite[Section 5.3]{AH09}, and \cite[Chapter 3]{Hig08}.

Let $\mathcal{U},\mathcal{W}$ be real Banach spaces.
The {Fr\'{e}chet derivative} of a function $\bm{f}:\mathcal{U} \rightarrow \mathcal{W}$ at a point $\bm{X}\in\mathcal{U}$, if it exists\footnote{We assume the functions considered in the paper are Fr\'{e}chet differentiable. We refer readers to works, such as \cite{Pel85,Bic12} for the conditions under which  a function is Fr\'{e}chet differentiable. }, is a unique linear mapping $\mathsf{D}\bm{f}[\bm{X}]:\mathcal{U}\rightarrow\mathcal{W}$ such that
\[
\frac{\left\|\bm{f}(\bm{X}+\bm{E})-\bm{f}(\bm{X})-\mathsf{D}\bm{f}[\bm{X}](\bm{E})\right\|_{\mathcal{W}} }{\|\bm{E}\|_{\mathcal{U}} }\rightarrow 0 \quad \text{as } \bm{E}\in\mathcal{U} \,\text{ and }\, \|\bm{E}\|_{\mathcal{U}}\rightarrow \bm{0}.
\]


The Fr\'{e}chet derivative also satisfies the sum rule, the product rule, and the chain rule as in the conventional derivatives of real-valued functions (see e.g. \cite[Theorem 3.4]{Hig08}).

The \emph{partial Fr\'{e}chet derivative} of multivariate functions can be defined as follows \cite[Section 5.3]{AH09}.
Let $\mathcal{U},\mathcal{V}$ and $\mathcal{W}$ be real Banach spaces, $\bm{f}:\mathcal{U}\times \mathcal{V} \rightarrow \mathcal{W}$. 
For a fixed $\bm{v}_0 \in\mathcal{V}$, $\bm{f}(\bm{u},\bm{v}_0)$ is a function of $\bm{u}$ whose derivative at $\bm{u}_0$, if it exists, is called the partial Fr\'{e}chet derivative of $\bm{f}$ with respect to $\bm{u}$, and is denoted by $\mathsf{D}_{\bm{u}}\bm{f} [\bm{u}_0, \bm{v}_0]$.
The partial Fr\'{e}chet derivative $\mathsf{D}_{\bm{v}}\bm{f} [\bm{u}_0, \bm{v}_0]$ is defined similarly.

For any map $\bm{f}:\mathcal{U}\rightarrow \mathcal{W}$ and an operator $\bm{X}\in\mathcal{U}$, we define the induced norm of the Fr\'{e}chet derivative $\mathsf{D}\bm{f}[\bm{X}]$ as
\begin{equation} \label{eq:Frechet_norm}
\left\| \mathsf{D}\bm{f}[\bm{X}] \right\| \triangleq  \sup_{\bm{E}\neq \bm{0} } \frac{ \left\| \mathsf{D}\bm{f}[\bm{X}](\bm{E}) \right\| }{\left\|\bm{E}\right\|},
\end{equation}
where the norm can be any consistent norm (e.g.~$\left\| \mathsf{D}\bm{f}[\bm{X}] \right\|_2 = \sup_{\bm{E}\neq \bm{0} }  \left\| \mathsf{D}\bm{f}[\bm{X}](\bm{E}) \right\|_2 /  \left\|\bm{E}\right\|_2$).


\subsection{Standard Matrix Functions} \label{Functions}
For each self-adjoint and bounded operator $\bm{A}\in\mathbb{M}^\text{sa}$ 
with the spectral decomposition
$\bm{A} = \int_{\lambda \in \sigma(\bm{A})} \lambda \, \mathrm{d} \bm{E}(\lambda)$,
we define the \emph{standard matrix function} of each scalar function by $f(\bm{A}) \triangleq  \int_{\lambda \in \sigma(\bm{X})} f(\lambda) \, \mathrm{d} \bm{E}(\lambda)$.
Note that we use lowercase Roman and Greek letters to denote standard matrix functions, while the calligraphic capital letter $\bm{f}$ refers to general operator-valued functions that are not necessarily standard.


A function $f:\mathbb{R}\rightarrow \mathbb{R}$ is called \emph{operator convex} if for each $\bm{A},\bm{B}\in \mathbb{M}^\text{sa}$ and $0\leq t \leq 1$,
$
f( t\bm{A} + (1-t)\bm{B}) \preceq t f(\bm{A}) + (1-t) f(\bm{B}).
$
Similarly, a function $f$ is called \emph{operator monotone} if, for each $\bm{A}, \bm{B} \in\mathbb{M}^\text{sa}$,
$
\bm{A}\preceq \bm{B} \; \Rightarrow \; f(\bm{A})\preceq f(\bm{B}).
$

The standard matrix function can be extended into the multivariate case by considering $n$-tuples commuting self-adjoint operators, i.e.~ 
\begin{align} \label{eq:multi}
f(\vec{\bm{X}}) \triangleq  \int_{(\lambda_1,\ldots,\lambda_n)} f(\lambda_1,\ldots,\lambda_n) \, \mathrm{d} \bm{E}(\lambda_1,\ldots,\lambda_n).
\end{align}
for $\vec{\bm{X}}=(\bm{X}_1,\ldots,\bm{X}_n) \in \mathbb{S}^n$ with $\bm{X}_i=\int_{\lambda} \lambda \, \mathrm{d} \bm{E}_i(\lambda)$ and $\bm{E}(\lambda_1,\ldots,\lambda_n)\triangleq  \bm{E}(\lambda_1)\cdots\bm{E}(\lambda_n)$.
Define the (first-order) divided difference of the multivariate matrix function $f:\mathbb{R}^n\rightarrow\mathbb{R}$, for each $i=1,\ldots,n$, by using the rule:
\begin{align} \label{eq:difference}
\varphi_i(\bar{x},\bar{y})= \frac{\left(f(\bar{x}-\bar{y})\right)(x_i-y_i)}{\|\bar{x}-\bar{y}\|_2^2} \; \text{for } \bar{x}\neq \bar{y} \quad \text{and} \quad \varphi_k(\bar{x},\bar{x}) = \frac{\partial f}{\partial x_i} \left( \bar{x} \right),
\end{align}
where $\bar{x}\triangleq (x_1,\ldots,x_n)\in\mathbb{R}^n$. 
A multivariate extension of the Dalecki\u{\i} and Kre\u{\i}n formula \cite[Theorem 3.11]{Hig08} gives an alternative formula for the partial Fr\'echet derivative of a standard matrix function:
\begin{equation} \label{eq:multi_DK}
\mathsf{D}_{\bm{X}_i} f \left[\vec{\bm{X}}\right] (\vec{\bm{E}}) = \left[ \varphi_i \left(\bar{\lambda}_k,\bar{\lambda}_l \right)\right]_{kl} \odot \bm{E}_i, \quad \forall \,\vec{\bm{X}}\in\mathbb{S}^n \text{ and } \vec{\bm{E}} = (\bm{E}_1,\ldots, \bm{E}_n) \in \left( \mathbb{M}^{\text{sa}}\right)^n,
\end{equation}
where the $\{\bm{X}_i\}_{i=1}^n$ are simultaneously diagonalised with $\lambda_{ik}$ being the $k$-th eigenvalue of $\bm{X}_i$; $\bar{\lambda}_k \triangleq  (\lambda_{1k},\ldots,\lambda_{ik},\ldots,\lambda_{nk})$; and $\odot$ denotes the Hadamard product.


\section{Matrix Functional Inequalities} \label{APP}

The main results in this section include various matrix functional inequalities for matrix-valued processes.

We first present the definition of the matrix $\Phi$-entropies for random matrices from Chen and Tropp \cite{CT14}, and provide the notation for the Dirichlet forms of matrix-valued functions.
Then, we derive the \emph{matrix Poincar\'{e} inequalities} for the general multivariate matrix-valued function $\bm{f}:\left(\mathbb{M}_{d_1}^\text{sa}\right)^n \rightarrow \mathbb{M}_{d_2}^\text{sa}$ (Theorem \ref{theo:Poincare}), multivariate standard matrix functions (Corollary \ref{coro:Poincare_multi})
in Section \ref{Efron}.
We then extend the matrix Poincar\'{e} inequality to Gaussian distribution (called \emph{matrix Gaussian Poincar\'{e} inequality}, Theorem \ref{theo:Gaussian}). 
These results rely on the \emph{matrix Efron-Stein inequality}, which was first proven in our previous work \cite[Theorem 5.1]{CH16}. 

Section \ref{Sobolev} presents the results on Sobolev inequalities for matrix $\Phi$-entropies. 
The \emph{matrix $\mathnormal{\Phi}$-Sobolev inequality} of symmetric Bernoulli random variables and that of  Gaussian random variables are in Theorem \ref{theo:Sobolev} and  Theorem \ref{theo:Sobolev_Gau}, respectively.
The \emph{matrix logarithmic $\Phi$-Sobolev inequalities} of symmetric Bernoulli random variables and Gaussian random variables are given in Corollaries \ref{coro:Log} and \ref{coro:Log_Gau}.

Throughout this section, let $\vec{\bm{X}}\triangleq (\bm{X}_1,\ldots,\bm{X}_n)$ be a series of independent random variables taking values in some Polish space and let $\bm{Z}\triangleq \bm{f}(\vec{\bm{X}})\in\mathbb{M}_d^{\text{sa}}$ be a random matrix such that $\|\mathbb{E}\bm{Z}^2\|_\infty<\infty$.
Let $\bm{X}_i'$ be an independent copy of $\bm{X}_i$, for $1\leq i\leq n$,  and denote by $\widetilde{\bm{X}}^{(i)}\triangleq (\bm{X}_1,\ldots,\bm{X}_{i-1},\bm{X}_i',\bm{X}_{i+1},\ldots,\bm{X}_n)$, i.e.~replacing the $i$-th component of $\vec{\bm{X}}$ by the independent copy $\bm{X}_i'$.
Let $\bm{X}_{-i}\triangleq (\bm{X}_1,\ldots,\bm{X}_{i-1},\bm{X}_{i+1}
,\ldots,\bm{X}_n)$ and $\mathbb{E}_i[\;\cdot\;]\triangleq \mathbb{E}[\;\cdot\;|\bm{X}_{-i}]$, i.e.~expectation with respect to the $i$-th variable. Finally, denote by $\bm{Z}_i'\triangleq  \bm{f} (\widetilde{\bm{X}}^{(i)} )$ for every $i=1,\ldots,n$.

For a convex function $\Phi:\mathbb{R}\to\mathbb{R}$ such that  $\mathbb{E}\|\bm{Z}\|_\infty<\infty$ and $\mathbb{E}\|\mathrm{\Phi}(\bm{Z})\|_\infty<\infty$,
The matrix $\mathrm{\Phi}$-entropy $H_\mathrm{\Phi}(\bm{Z})$ is defined as
\[ 
H_\mathrm{\Phi}(\bm{Z})\triangleq \tr\left[\mathbb{E}\mathrm{\Phi}(\bm{Z})-\mathrm{\Phi}(\mathbb{E}\,\bm{Z})\right].
\]

The (matrix-valued) \emph{local influence} is defined as
\[
\bm{\mathcal{I}}_i(\bm{f}|\bm{X}_{-i}) \triangleq  \frac12 \mathbb{E}\left[ \left( \bm{f}(\vec{\bm{X}})-\bm{f}\left(\widetilde{\bm{X}}^{(i)}\right)\right)^2\right],
\]
for every $i\in1,\ldots, n$. This quantity characterises the local fluctuation of the $i$-th position on average.
The matrix-valued and real-valued \emph{total influence} of the function $\bm{f}$ are defined as
\begin{align} \label{eq:engergy}
\bm{\mathcal{E}}(\bm{f}) \triangleq \sum_{i=1}^n \mathbb{E}\left[ \bm{\mathcal{I}}_i ( \bm{f}| \bm{X}_{-i})\right], \quad \text{and} \quad \mathcal{E}(\bm{f}) \triangleq \tr\bm{\mathcal{E}}(\bm{f}). 
\end{align}
Note that we use the notation $\bm{\mathcal{E}}(\bm{f})$ and $\bm{\mathcal{E}}(\bm{Z})$ interchangeably.
This quantity is also known as the \emph{Dirichlet form} or \emph{energy functional} in Markov semigroup theory (see e.g.~\cite{BGL13} and \cite{CHT15}).
The total influence $\mathcal{E}(\bm{Z})$ has the following equivalent expressions (Lemma \ref{lemm:plus}):  
\begin{align} 
\bm{\mathcal{E}}(\bm{Z}) = \frac12 \sum_{i=1}^n \mathbb{E}\left[\left(\bm{Z}-\bm{Z}_i'\right)^2\right]
= \sum_{i=1}^n \mathbb{E} \left[ \left( \bm{Z}-\mathbb{E}_i \bm{Z} \right)^2 \right] 
=\sum_{i=1}^n \mathbb{E}\left[\left(\bm{Z}-\bm{Z}_i'\right)_+^2\right], \label{eq:Efron_+}
\end{align}
where, for $f(x) = \max\{x,0\}$, $(\bm{A})_+=f(\bm{A})$ denotes the contribution from its positive eigenvalues.

\subsection{The Matrix Poincar\'{e} Inequality} \label{Efron}

We denote the matrix-valued and real-valued variance of a random matrix $\bm{A}$ (taking values in $\mathbb{M}_d^{sa}$) by
\[
\textbf{{Var}}(\bm{A}) \triangleq   \left[\mathbb{E}\left( \bm{A} - \mathbb{E}\bm{A} \right)^2\right] = \left[\mathbb{E}\bm{A}^2 - \left(\mathbb{E}\bm{A}\right)^2 \right], \quad \text{and}\quad
\text{Var} (\bm{A})\triangleq \tr \textbf{Var}(\bm{A}).
\]
In the following, we recall the  result from our previous work \cite{CH16} and Eq.~\eqref{eq:Efron_+} to  present a matrix Efron-Stein inequality, which plays a major role in the matrix Poincar\'e inequality. 



\begin{theo}
	[Matrix Efron-Stein Inequality {\cite[Theorem 5.1]{CH16}}] \label{theo:Efron}
	With the prevailing assumptions, we have
	\begin{equation}\label{eq:Efron}
	\textnormal{\textbf{Var}}(\bm{Z}) \preceq \bm{\mathcal{E}}(\bm{Z})= \sum_{i=1}^n \mathbb{E}\left[\left(\bm{Z}-\bm{Z}_i'\right)_+^2\right].
	\end{equation}
\end{theo}

\begin{remark}
	In addition to Eq.~(\ref{eq:Efron_+}), the quantity $\bm{\mathcal{E}}(\bm{Z})$ also has the following minimum representation:
	\begin{align} \label{eq:Efron_min} 
	\bm{\mathcal{E}}(\bm{Z})=  \sum_{i=1}^n \min_{\bm{Z}_i} \mathbb{E} \left[\left(\bm{Z}-\bm{Z}_i\right)^2\right],
	\end{align}
	where the minimum is taken over the class of all $(\bm{X}_{-i})$-measurable random matrix $\bm{Z}_i$ such that $\|\mathbb{E}\bm{Z}_i^2\|_\infty<\infty$.  The identity follows from the fact that
	$
	\textnormal{\textbf{Var}}\left(\bm{A} \right) = \min_{\bm{u}\in\mathbb{M}_d^{sa}} \mathbb{E}\left[\left(\bm{A-u}\right)^2\right],
	$ for any random matrix $\bm{A}$ taking values in $\mathbb{M}_d^{sa}$.
	Therefore, for every $i=1,\ldots,n$,
	$
	\textnormal{Var}^{(i)}\left(\bm{Z} \right) = \min_{\bm{u}}  \mathbb{E}_i\left[\left(\bm{Z}-\bm{u}\right)^2\right],
	$
	where the infimum is taken over the class of all $(\bm{X}_{-i})$-measurable and square-integrable matrices $\bm{u}$. Note that the minimum is attained as $\bm{u}=\mathbb{E}_i\bm{Z}$.
	\Endremark
\end{remark}


The matrix Efron-Stein inequality can be used to prove a matrix version of the Poincar\'{e} inequality.

\begin{theo}[Matrix Poincar\'{e} Inequality] \label{theo:Poincare}
	Let  $\vec{\bm{X}}=(\bm{X}_1,\ldots,\bm{X}_n)\in (\mathbb{M}_{d_1}^\textnormal{sa})^n$ be an $n$-tuple independent random matrix taking values in the interval $[\bm{0},\bm{I}]$ (under the L\"{o}wner partial ordering) and let $\bm{f}:(\mathbb{M}_{d_1}^\textnormal{sa}([0,1]))^n\rightarrow \mathbb{M}_{d_2}^\textnormal{sa}$ be a separately  convex function\footnote{Note that $\bm{f}$ here is a multivariate super-operator. The separate convexity means that: for $ 0\leq t\leq 1$, 
		\[
		t \bm{f}\left(\bm{Y}\right) + (1-t) \bm{f}\left(\widetilde{\bm{Y}}^{(i)} \right) \preceq \bm{f}\left( t\vec{\bm{Y}} + (1-t) \widetilde{\bm{Y}}^{(i)} \right)
		\]
		for $\vec{\bm{Y}}=(\bm{Y}_1,\ldots,\bm{Y}_n)\in \left(\mathbb{M}_d^{sa}\right)^n $, and $\widetilde{\bm{Y}}^{(i)} = (\bm{Y}_1,\ldots,\bm{Y}_{i-1},\bm{Y}_i',\bm{Y}_{i+1},\ldots,\bm{Y}_n)\in\left(\mathbb{M}_d^{sa}\right)^n$ .
		The separate monotonicity is defined similarly.
	}		
	with finite partial Fr\'{e}chet derivatives.
	Then $\bm{f}(\vec{\bm{X}})=\bm{f}(\bm{X}_1,\ldots,\bm{X}_n)$ satisfies
	\begin{align} \label{eq:Poincare}
	\textnormal{Var}\left(\bm{f}\left(\vec{\bm{X}}\right)\right) \leq \sum_{i=1}^n \mathbb{E} \left[ \left\| \mathsf{D}_{\bm{X}_i}\bm{f} \left[\vec{\bm{X}}\right]\right\|_2^2 \right],
	\end{align}
	where $\| \mathsf{D}_{\bm{X}_i}\bm{f} [\vec{\bm{X}} ] \|_2$ is the norm of the Fr\'{e}chet derivative defined in Eq.~\eqref{eq:Frechet_norm}.
	
	Moreover, if $\bm{f}$ is separately monotone decreasing, then Eq.~\eqref{eq:Poincare} can be strengthened to
	\begin{align} \label{eq:Poincare2}
	\textnormal{\textbf{Var}}\left(\bm{f}\left(\vec{\bm{X}}\right)\right) 
	\preceq_\textnormal{U} \sum_{i=1}^n \mathbb{E} \left[  \left(  \mathsf{D}_{\bm{X}_i}\bm{f} \left[\vec{\bm{X}}\right](\bm{I}) \right)^2 \right].
	\end{align} 
\end{theo}

\begin{proof}
	Recall $\bm{Z}\equiv\bm{f}\left(\vec{\bm{X}}\right)$ and $\bm{Z}_i\equiv\bm{f}\left(\widetilde{\bm{X}}^{(i)}\right)=\bm{f}(\bm{X}_1,\ldots,\bm{X}_{i-1},\bm{X}_i',\bm{X}_{i+1},\ldots,\bm{X}_n)$, where $\bm{X}_i'$ is an identical copy of $\bm{X}_i$. The proof follows from the matrix Efron-Stein inequality (Theorem \ref{theo:Efron}):
	\begin{align} \label{eq:temp_Poincare}
	\textbf{Var}\left(\bm{f}\left(\vec{\bm{X}}\right)\right) = \textbf{Var}\left(\bm{Z}\right) \preceq \sum_{i=1}^n \mathbb{E}  \left( \bm{Z}- \bm{Z}_i \right)_+^2.
	\end{align}
	It then suffices to bound each term $\mathbb{E}  \left( \bm{Z}-\bm{Z}_i \right)_+^2$ on the right-hand side above:
	\begin{align}
	\left(\bm{Z}- \bm{Z}_i \right)_+^2 
	=  \left( \bm{f}\left(\vec{\bm{X}}\right) - \bm{f}\left(\widetilde{\bm{X}}^{(i)}\right) \right)_+^2 
	&\preceq_\textnormal{U}  \left( \mathsf{D}_{\bm{X}_i} \bm{f}\left[\vec{\bm{X}}\right]\left( \bm{X}_i - \bm{X}_i' \right) \right)_+^2 \label{eq:temp_Poincare3} \\
	 &\preceq_\text{U}  \left( \mathsf{D}_{\bm{X}_i} \bm{f}\left[\vec{\bm{X}}\right]\left( \bm{I} \right) \right)^2, \label{eq:temp_Poincare2}
	\end{align}
	where the first inequality follows from the separate convexity (see e.g. \cite[Chapter 3]{Fle80}):
		\begin{equation}\label{eq_separate}
		\bm{f}\left(\vec{\bm{X}}\right) - \bm{f}\left(\widetilde{\bm{X}}^{(i)}\right) \preceq 
		\mathsf{D}_{\bm{X}_i} \bm{f}\left[\vec{\bm{X}}\right]\left( \bm{X}_i - \bm{X}_i' \right),
		\end{equation}
	as well as the fact that $\bm{A}_+\preceq\bm{B}_+$ implies $\bm{A}_+^2\preceq_\text{U}\bm{B}_+^2$.
	The second line is due to the separately monotone decreasing property: 
	$\mathsf{D}_{\bm{X}_i} \bm{f}[\vec{\bm{X}}]\left( \bm{A}\right) \preceq 0$ for all $\bm{A}\succeq 0$.	
	Hence $( \mathsf{D}_{\bm{X}_i} \bm{f}[\vec{\bm{X}}]( \bm{X}_i - \bm{X}_i' ) )_+ \preceq   -\mathsf{D}_{\bm{X}_i} \bm{f}[\vec{\bm{X}}]( \bm{I} ) $, which proving Eq.~\eqref{eq:Poincare2}.
	
	If $\bm{f}$ is not separately monotone decreasing, Eq.~\eqref{eq:temp_Poincare2} does not generally hold. 
	To show Eq.~\eqref{eq:Poincare}, we take the normalised trace on both sides of Eq.~\eqref{eq:temp_Poincare3}:
	\begin{align*}
	\tr \left(\bm{Z}- \bm{Z}_i \right)_+^2 
	&\leq \tr \left( \mathsf{D}_{\bm{X}_i} \bm{f}\left[\vec{\bm{X}}\right]\left( \bm{X}_i - \bm{X}_i' \right) \right)_+^2
	\leq \tr \left| \mathsf{D}_{\bm{X}_i} \bm{f}\left[\vec{\bm{X}}\right]\left( \bm{X}_i - \bm{X}_i' \right) \right|^2\\
	&=  \frac{1}{d_2}\left\| \mathsf{D}_{\bm{X}_i} \bm{f}\left[\vec{\bm{X}}\right]\left( \bm{X}_i - \bm{X}_i' \right)  \right\|_2^2 
	\leq \frac{1}{d_2}\left\| \mathsf{D}_{\bm{X}_i}\bm{f}\left[\vec{\bm{X}}\right]\right\|_2^2 \cdot \left\|\bm{X}_i-\bm{X}_i' \right\|_2^2\\
	&\leq \frac{1}{d_2}\left\| \mathsf{D}_{\bm{X}_i}\bm{f}\left[\vec{\bm{X}}\right]\right\|_2^2 \cdot \left\| \bm{I}  \right\|_2^2
	= \left\| \mathsf{D}_{\bm{X}_i}\bm{f}\left[\vec{\bm{X}}\right]\right\|_2^2.
	\end{align*}
	The equality in the second line follows from the definition of Schatten $2$-norm.
	The inequality in the second line follows directly from the norm of Fr\'{e}chet derivatives, i.e.~
	$\left\| \mathsf{D} f[\bm{A}](\bm{B})\right\|_2
	\leq \left\| \mathsf{D} f[\bm{A}]\right\|_2 \cdot \| \bm{B} \|_2$.
	Finally, we use the assumption $\bm{0} \preceq \bm{X}_i,\bm{X}_i'\preceq \bm{I}$ and $\|\bm{I}\|_2=\sqrt{d_2}$ in the third line to complete the proof.
\end{proof}

Note that Theorem \ref{theo:Poincare} generalises the classical Poincar\'{e} inequality (e.g.~\cite[Theorem 3.17]{BLM13}):
\[
\textnormal{Var}(f(X)) \leq \mathbb{E}\left[\left\|\nabla f(X) \right\|^2 \right],
\]
where $X\triangleq (X_1,\ldots,X_n)$ denotes an independent random vector and each element takes values in the interval $[0,1]$.


Theorem \ref{theo:Poincare} considers the matrix Poincar\'{e} inequality for general matrix functions $\bm{f}:\left(\mathbb{M}_d^\text{sa}\right)^n\rightarrow \mathbb{M}_d^\text{sa}$, while in Corollary \ref{coro:Poincare_multi} below, 
we impose additional pairwise commutative criteria on $\vec{\bm{X}}=(\boldsymbol{X}_1,\ldots,\boldsymbol{X}_n)$; namely, $[\bm{X}_i, \bm{X}_j]=\bm{0}$ almost surely for $i\neq j \in[n]$. We establish the following Matrix Poincar\'e inequality for \emph{multivariate standard matrix functions (see~Eq.~\eqref{eq:multi})}.
\begin{coro}
	[Matrix Poincar\'{e} Inequality for Multivariate Standard Matrix Functions] \label{coro:Poincare_multi}
	Let  $\vec{\bm{X}}=(\boldsymbol{X}_1,\ldots,\boldsymbol{X}_n)$ be an $n$-tuple independent random matrix taking values in $\mathbb{S}_d^n$ with joint spectrum in $[0,1]^n$.
	Let $f:([0,1])^n\rightarrow \mathbb{R}$ be a multivariate standard matrix function that is separately operator convex and has finite partial Fr\'{e}chet derivatives.
	Then, $f\left(\vec{\bm{X}}\right)=f(\boldsymbol{X}_1,\ldots,\boldsymbol{X}_n)$ satisfies
	\[
	\textnormal{\textbf{Var}}\left(f\left(\vec{\bm{X}}\right)\right) \preceq_\textnormal{U} \sum_{i=1}^n \mathbb{E}  \left[ \left( \left[ \varphi_i \left(\bar{\lambda}_k,\bar{\lambda}_\ell \right)\right]_{k\ell} \odot \bm{I} \right)^2 \right],
	\]	
	where $\varphi_i$ is the divided difference of $f$ defined in Eq.~\eqref{eq:difference}, and $\bar{\lambda}_k\triangleq (\lambda_{1k}, \ldots, \lambda_{ik}, \ldots, \lambda_{nk})$ with $\lambda_{ik}$ being the $i$-th eigenvalue of $\bm{X}_k$.
\end{coro}

\begin{proof}
	Following the same argument in the proof of Theorem \ref{theo:Poincare}, we have
	\begin{align*}
	\mathsf{D}_{\bm{X}_i} f[\vec{\bm{X}}]\left( \bm{X}_i - \bm{X}_i' \right)   
	= \left[ \varphi_i \left(\bar{\lambda}_k,\bar{\lambda}_\ell \right)\right]_{k\ell} \odot \left(\bm{X}_i-\bm{X}_i'\right) \preceq  \left[ \varphi_i \left(\bar{\lambda}_k,\bar{\lambda}_\ell \right)\right]_{k\ell} \odot \bm{I},
	\end{align*}	
	where the equality follows from the multivariate version of Dalecki\u{\i} and Kre\u{\i}n formula \eqref{eq:multi_DK}, and the inequality is a direct consequence of Schur's product theorem: $\bm{A}\odot \bm{B} \leq \bm{A}\odot \bm{I} \cdot \|\bm{B}\|_\infty$ \cite{Sch11}.
\end{proof}

The matrix Efron-Stein inequality is used in Theorem \ref{theo:Poincare} to prove the matrix Poincar\'{e} inequality.
Next we will show that the matrix Efron-Stein inequality can be also applied to establish an upper bound, known as the \emph{Gaussian Poincar\'{e} inequality}, for a Fr\'{e}chet differentiable matrix-valued function of \emph{Gaussian Unitary Ensembles (GUE)}\footnote{The Gaussian Unitary Ensembles are a family of random Hermitian matrices whose upper-triangular entries are independently and identically distributed (i.i.d.)~complex standard Gaussian random variables, while the diagonal entries are i.i.d.~real standard Gaussian random variables, see e.g.~\cite[\S 2.6]{Tao12}}.

\begin{theo}[Matrix Poincar\'{e} Inequality for GUE] \label{theo:Gaussian}
	Let $\bm{X}_1,\ldots,\bm{X}_n$ be independent random matrices\footnote{We consider ``entry-wise" independence here.} from the Gaussian Unitary Ensemble
	and let $\bm{f}:\left(\mathbb{M}_{d_1}^\textnormal{sa}\right)^n\rightarrow \mathbb{M}_{d_2}^\textnormal{sa}$ be any twice Fr\'{e}chet differentiable function.
	Then $\bm{f}\left(\vec{\bm{X}}\right)$ satisfies
	\[
	\textnormal{Var}\left(\bm{f}\left(\vec{\bm{X}}\right)\right) 
	\leq \sum_{i=1}^n \mathbb{E} \left[ \big\| \mathsf{D}_{\bm{X}_i}\bm{f}\left[\vec{\bm{X}}\right]\big\|_2^2 \right].
	\]	
\end{theo}
\begin{proof}
	Borrowing the idea from \cite{ABC+00} (see also \cite[Theorem 3.20]{BLM13}) to prove this theorem.
	We first assume $\sum_{i=1}^n \mathbb{E} \| \mathsf{D}_{\bm{X}_i}\bm{f}[\vec{\bm{X}}]\|_2^2 < \infty$; otherwise the inequality trivially holds.
	Then, it suffices to establish this theorem for $n=1$:
	\begin{align} \label{eq:Gaussian1}
	\textnormal{Var}\left(\bm{f}(\bm{X})\right) \leq  \mathbb{E} \left[ \big\| \mathsf{D}\bm{f}[\bm{X}]\big\|_2^2 \right],
	\end{align}
	and it can be easily extended to every $n\in\mathbb{N}$ by applying the matrix Efron-Stein inequality, Theorem \ref{theo:Efron}.
	
	Now, for every $j\in[m]\triangleq\{1,\ldots,m\}$, denote by $\bm{W}_j$, $\bm{W}_j'$ the $d_1\times d_1$ matrices whose entries are sampled from independent Rademacher random variables (i.e.~uniformly $\{\pm 1\}$-valued random variables). Let
	\[
	\bm{Y}_j = \frac{\ \left(\bm{W}_j+\mathrm{i}\cdot \bm{W}_j'\right) + \left(\bm{W}_j+\mathrm{i}\cdot \bm{W}_j'\right)^\dagger }{2}.
	\]
	Denote by $\epsilon_1,\ldots,\epsilon_m$ a series of independent Rademacher random variables, and define
	$\bm{S}_m \triangleq \frac{1}{\sqrt{m}} \sum_{j=1}^m \epsilon_j \bm{Y}_j $.
	Then, for every $j\in [m]$,
	\begin{align*}
	\text{Var}^{(j)} \left(\bm{f}(\bm{S}_m) \right)
	&= \frac14 \tr \mathbb{E}_{\bm{Y}_j} 
	\left[ \left( \bm{f}\left(\bm{S}_m + \frac{1-\epsilon_j}{\sqrt{m}} \bm{Y}_j \right) - \bm{f}\left(\bm{S}_m - \frac{1+\epsilon_j}{\sqrt{m}} \bm{Y}_j \right) \right)^2 \right].
	\end{align*}
	We invoke the matrix Efron-Stein inequality to obtain
	\begin{align} \label{eq:Gaussian2}
	\text{Var} \left(\bm{f}(\bm{S}_m) \right) 
	\leq \frac14 \sum_{j=1}^m \tr \mathbb{E} \left[ \left( \bm{f}\left(\bm{S}_m + \frac{1-\epsilon_j}{\sqrt{m}} \bm{Y}_j \right) - \bm{f}\left(\bm{S}_m - \frac{1+\epsilon_j}{\sqrt{m}} \bm{Y}_j \right) \right)^2 \right],
	\end{align}
	 and then use Taylor's expansion to further bound the right-hand side of Eq.~\eqref{eq:Gaussian2}. For every $i\in[n]$ and some constants $0\leq \alpha,\beta \leq 1$, it follows almost surely that 
	\begin{align*}
	\bm{f}\left(\bm{S}_m + \frac{1-\epsilon_j}{\sqrt{m}} \bm{Y}_j \right) &= \bm{f}(\bm{S}_m) + \mathsf{D}\bm{f}[\bm{S}_m]\left( \frac{1-\epsilon_j}{\sqrt{m}} \bm{Y}_j \right) + \mathcal{R}_2\left( \bm{S}_m, \frac{1-\epsilon_j}{\sqrt{m}} \bm{Y}_j \right);\\
	\bm{f}\left(\bm{S}_m - \frac{1+\epsilon_j}{\sqrt{m}} \bm{Y}_j \right) &= \bm{f}(\bm{S}_m) + \mathsf{D}\bm{f}[\bm{S}_m]\left( -\frac{1+\epsilon_j}{\sqrt{m}} \bm{Y}_j \right) + \mathcal{R}_2\left( \bm{S}_m, - \frac{1+\epsilon_j}{\sqrt{m}} \bm{Y}_j \right),
	\end{align*}
	where $\mathcal{R}_l:\mathbb{M}_{d_1} \times \mathbb{M}_{d_1} \rightarrow \mathbb{M}_{d_2}$ is the remainder term of the Taylor series:
	\begin{align*}
	\mathcal{R}_l(\bm{X},\bm{E}) \triangleq 
	\sum_{k=l}^\infty\frac{1}{k!}\mathsf{D}^{k}\bm{f}\left[\bm{X}\right](\underbrace{\bm{E},\ldots,\bm{E}}_k) 
	= o\left( |\bm{E}|^l \right).
	\end{align*}
	Therefore,
	\begin{align*}
	\bm{f}\left(\bm{S}_m + \frac{1-\epsilon_j}{\sqrt{m}} \bm{Y}_j \right) - 	\bm{f}\left(\bm{S}_m - \frac{1+\epsilon_j}{\sqrt{m}} \bm{Y}_j \right) 
	&\preceq \frac{2}{\sqrt{m}} \mathsf{D} \bm{f} [\bm{S}_m] \left( \bm{Y}_j \right) 
	+ o\left( \frac1m \right),
	\end{align*}
	and 
	\begin{align*}
	&\frac{1}{4} \sum_{j=1}^m\tr \mathbb{E} \left[ \left( \bm{f}\left(\bm{S}_m + \frac{1-\epsilon_j}{\sqrt{m}} \bm{Y}_j \right) - 	\bm{f}\left(\bm{S}_m - \frac{1+\epsilon_j}{\sqrt{m}} \bm{Y}_j \right) \right)^2 \right]
	\leq \big\| \mathsf{D}\bm{f}[\bm{S}_m] \big\|_2^2 
	+ o\left( \frac1{\sqrt{m}} \right).
	\end{align*}
	Let $m$ go to infinity, we have
	\begin{align} \label{eq:Gaussian3}
	\lim_{m\rightarrow \infty} \frac14 \sum_{i=1}^m \tr \mathbb{E} \left[ \left( \bm{f}\left(\bm{S}_m + \frac{1-\epsilon_j}{\sqrt{m}} \bm{Y}_j \right) - \bm{f}\left(\bm{S}_m - \frac{1+\epsilon_j}{\sqrt{m}} \bm{Y}_j \right) \right)^2 \right]
	\leq \mathbb{E} \left[  \big\| \mathsf{D}\bm{f}[\bm{X}] \big\|_2^2 \right],
	\end{align}
	where, by the central limit theorem (see Lemma \ref{lemm:GUE}), $\bm{S}_m$ converges in distribution to a random matrix $\bm{X}$ in GUE. Thus $\text{Var} \left(\bm{f}(\bm{S}_m) \right)$ converges to $\text{Var} \left(\bm{f}(\bm{X}) \right)$.
	
	Finally, the subadditivity of the variance and Eq.~\eqref{eq:Gaussian1} lead to
	\begin{align*}
	\text{Var}\left(\bm{f}\left(\vec{\bm{X}}\right)\right) 
	\leq \sum_{i=1}^n \mathbb{E} \left[ \text{Var}^{(i)} \big(\bm{f}\left(\vec{\bm{X}}\right)\big) \right]
	\leq \sum_{i=1}^n \mathbb{E} \Big[   \mathbb{E}_i \left[ \big\| \mathsf{D}_{\bm{X}_i}\bm{f}\left[\vec{\bm{X}}\right]\big\|_2^2 \right] \Big]
	= \sum_{i=1}^n \mathbb{E} \left[ \big\| \mathsf{D}_{\bm{X}_i}\bm{f}\left[\vec{\bm{X}}\right]\big\|_2^2 \right],
	\end{align*}
	which completes the proof.
\end{proof}

\subsection{Matrix $\mathrm{\Phi}$-Sobolev Inequalities} \label{Sobolev}


In this section, we consider matrix-valued functions defined on Boolean hypercubes: $\boldsymbol{f}:\{0,1\}\rightarrow \mathbb{M}_d^\text{sa}$ and establish matrix $\mathrm{\Phi}$-Sobolev inequalities.
The main ingredient to prove this inequality comes from Fourier analysis and the hypercontractive inequality for matrix-valued functions.

Ben-Aroya {et al.} \cite{BRW08} generalised Bonami and Beckner's results by considering matrix-valued functions $\boldsymbol{f}:\{0,1\}\rightarrow \mathbb{M}_d$.
Similarly, Fourier analysis can be naturally extended into the matrix setting; that is, the Fourier transform $\widehat{\boldsymbol{f}}$ of the matrix-valued function $\bm{f}$ can be expressed as
\begin{align*}
\begin{cases}
\widehat{\boldsymbol{f}}(S) = \frac{1}{2^n} \sum_{ x\in\{0,1\}^n } \boldsymbol{f}(x)u_S(x);\\
\boldsymbol{f}(x) = \sum_{S \subseteq \{1,\ldots,n\} } \widehat{\boldsymbol{f}}(S) u_S(x),
\end{cases}
\end{align*}
where $u_S(x) \triangleq  \Pi_{i\in S} (-1)^{x_i}$.
Therefore, the classical hypercontractive inequality \cite{Bon70,Bec75} can be extended to matrix-valued functions.

\begin{theo}
	[Matrix Bonami-Gross-Beckner Inequality \cite{BRW08}] \label{theo:matrix_hyper}
	For every $\boldsymbol{f}:\{0,1\}^n \rightarrow \mathbb{M}_d$ and $1\leq p \leq 2$,
	\[
	\left( \sum_{S\subseteq [n]} (p-1)^{|S|} \left\| \widehat{\boldsymbol{f}}(S)\right\|_{p}^2 \right)^{1/2} 
	\leq 
	\left( \frac{1}{2^n} \sum_{ x\in\{0,1\}^n } \left\| \boldsymbol{f}(x)\right\|_{p}^p \right)^{1/p}.	
	\]
\end{theo}

With Theorem~\ref{theo:matrix_hyper}, we can prove a matrix $\mathrm{\Phi}$-Sobolev inequality for matrix-valued functions defined on symmetric Bernoulli random variables.
\begin{theo}[Matrix $\mathrm{\Phi}$-Sobolev Inequalities for Symmetric Bernoulli Random Variables] \label{theo:Sobolev}
	Let $X$ be uniformly distributed over $\mathcal{X}\equiv\{0,1\}^n$ (an $n$-dimensional binary hypercube) and $\boldsymbol{f}:\mathcal{X} \rightarrow \mathbb{M}_d^+$ be an arbitrary matrix-valued function. Then for all $p\in (1,2)$, and $\mathrm{\Phi}(u)=u^{2/p}$,
	\begin{equation} \label{eq:Sobolev}
	H_\mathrm{\Phi} (\boldsymbol{f}^p) \leq (2-p)\mathcal{E}(\boldsymbol{f})\cdot d^{1-2/p} + \tr \mathbb{E}[\boldsymbol{f}^2]\cdot(1-d^{1-2/p}).
	\end{equation}
\end{theo}
\begin{proof}
	Starting from the left-hand side of Eq.~\eqref{eq:Sobolev}, the definition of the matrix $\Phi$-entropy functional gives
	\begin{align}
	H_\mathrm{\Phi} (\boldsymbol{f}^p) &= \tr \mathbb{E} \left[ \boldsymbol{f}^2 \right] - \tr \left[ \left( \mathbb{E} \boldsymbol{f}^p \right)^{2/p} \right]  
	\leq \tr \mathbb{E} \left[ \boldsymbol{f}^2 \right] - \left( \tr \mathbb{E} \boldsymbol{f}^p \right)^{2/p} \notag\\
	&= \tr \mathbb{E} \left[ \boldsymbol{f}^2 \right] - \left( \mathbb{E} \left\| \boldsymbol{f} \right\|_{p}^p \right)^{2/p} \cdot d^{-2/p}, \label{eq:Sobolev1}
	\end{align}
	where we apply the convexity of $\tr$, and recall that $(\;\cdot\;)^{2/p}$ is a convex function for $1\leq p \leq 2$.
	
	We then apply Theorem \ref{theo:matrix_hyper} to  Eq.~\eqref{eq:Sobolev1} to obtain
	\begin{align}
	H_\mathrm{\Phi}(\boldsymbol{f}^p)
	&\leq \tr \mathbb{E} \left[ \boldsymbol{f}^2 \right] - 	\left( \sum_{S\subseteq [n]} (p-1)^{|S|} \left\| \widehat{\boldsymbol{f}}(S)\right\|_{p}^2 \right) \cdot d^{-2/p}\notag \\
	&\leq \tr \mathbb{E} \left[ \boldsymbol{f}^2 \right] - 	\left( \sum_{S\subseteq [n]} (p-1)^{|S|} \tr \left[ \widehat{\boldsymbol{f}}(S)^2\right] \right)\cdot d^{1-2/p} \notag \\
	&= \tr \left[ \sum_{S\subseteq [n]} \widehat{\boldsymbol{f}}(S)^2 \right] - 	\left( \sum_{S\subseteq [n]} (p-1)^{|S|} \tr \left[ \widehat{\boldsymbol{f}}(S)^2\right] \right)\cdot d^{1-2/p} \notag \\
	&= \tr \left[ \sum_{S\subseteq [n]} \left(1-(p-1)^{|S|} d^{1-2/p} \right) \widehat{\boldsymbol{f}}(S)^2 \right], \label{eq:Sobolev3}
	\end{align}	
	where the second inequality is because the Schatten $p$-norm is non-increasing, and
	we apply Parseval's identity (Lemma \ref{lemm:Parseval} in Appendix \ref{app_lemmas}) in the third line.
	
	From the elementary analysis, it can be verified that, for all $S\subseteq [n]$ and $1\leq p\leq 2$,
$
	1-(p-1)^{|S|} \leq (2-p)|S|.
$
	Therefore, it follows that
	\[
	1-(p-1)^{|S|}d^{1-2/p} \leq (2-p)|S|d^{1-2/p} + (1-d^{1-2/p}).
	\]
	Finally, given $\sum_{ S\subseteq [n] } \tr \left[ |S| \widehat{\boldsymbol{f}}(S)^2 \right] = \mathcal{E}(\boldsymbol{f})$ (see Lemma \ref{lemm:Var} in Appendix \ref{app_lemmas}),  Eq.~\eqref{eq:Sobolev3} can be further deduced as
	\begin{align*}
	H_\mathrm{\Phi}(\boldsymbol{f}^p) 
	&\leq  \tr \left[ \sum_{S\subseteq [n]} \left(1-(p-1)^{|S|} d^{1-2/p} \right) \widehat{\boldsymbol{f}}(S)^2 \right] \\
	&\leq   \tr \left[ \sum_{S\subseteq [n]} \left( (2-p)|S|d^{1-2/p} + (1-d^{1-2/p})  \right) \widehat{\boldsymbol{f}}(S)^2 \right]\\
	&= (2-p)\mathcal{E}(\boldsymbol{f})\cdot d^{1-2/p} + \tr \mathbb{E}[\boldsymbol{f}^2]\cdot(1-d^{1-2/p}),
	\end{align*}
	which completes our claim.
\end{proof}

\begin{theo}[Matrix $\mathrm{\Phi}$-Sobolev Inequalities for Gaussian Distributions] \label{theo:Sobolev_Gau}
	Let $X=(X_1,\ldots,X_n)$ be a vector of $n$ independent standard Gaussian random variables taking values in $\mathcal{X}\equiv \mathbb{R}^n$, and 
	let $\boldsymbol{f}:\mathcal{X} \rightarrow \mathbb{M}_d^+$ be an arbitrary matrix-valued function. Then for all $p\in (1,2)$, and $\mathrm{\Phi}(u)=u^{2/p}$,
	\begin{equation} \label{eq:Sobolev_Gau}
	H_\mathrm{\Phi} (\boldsymbol{f}^p) \leq (2-p)\sum_{i=1}^n \mathbb{E} \left[ \big\| \mathsf{D}_{{X}_i}\bm{f}[{X}]\big\|_2^2 \right]\cdot d^{1-2/p} + \tr \mathbb{E}[\boldsymbol{f}^2]\cdot(1-d^{1-2/p}).
	\end{equation}
\end{theo}
\begin{proof}
	The proof parallels the approach in Theorem \ref{theo:Gaussian}.
	Recall from Eq.~\eqref{eq:Gaussian3} and let $\bm{Y}_i\equiv 1$:
	\[
	\mathcal{E}^{(i)}(\bm{f}) =
	\lim_{m\rightarrow \infty} \frac14 \sum_{i=1}^m \tr \mathbb{E}_i \left[ \left( \bm{f}\left({S}_m + \frac{1-\epsilon_i}{\sqrt{m}}  \right) - \bm{f}\left({S}_m - \frac{1+\epsilon_i}{\sqrt{m}}  \right) \right)^2 \right]
	= \mathbb{E}_i \left[  \big\| \mathsf{D}_{{X}_i} \bm{f}[{X}] \big\|_2^2 \right].
	\]
	This and Theorem~\eqref{theo:Sobolev} yield Eq.~\eqref{theo:Sobolev_Gau} and the statement follows.
\end{proof}

The logarithmic Sobolev inequality for matrix-valued functions immediately follows from Theorems \ref{theo:Sobolev} and \ref{theo:Sobolev_Gau}.
Subsequently, we denote by
\begin{align}
\Ent(\bm{Z}) := H_\Phi(\bm{Z}), \quad \text{when } \Phi(u) = u\log u.
\end{align}

\begin{coro}[Matrix Log-Sobolev Inequalities for Symmetric Bernoulli Random Variables]\label{coro:Log}
	Let $\boldsymbol{f}:\{0,1\}^n \rightarrow \mathbb{M}_d^+$ be an arbitrary matrix-valued function defined on the $n$-dimensional binary hypercube and assume that $X$ is uniformly distributed over $\{0,1\}^n$. Then
	\[
	\Ent (\bm{f}^2) \leq  2 \mathcal{E}(\bm{f}) + \log(d) \cdot \tr \mathbb{E}\left[\bm{f}^2\right].
	\]
\end{coro}
\begin{proof}
	By letting $p\rightarrow 2$, the left-hand side of Eq.~\eqref{eq:Sobolev} becomes
	\[
	\lim_{p\rightarrow 2^-} \frac{H_\Phi (\bm{f}^p)}{2-p} 
	= \lim_{p\rightarrow 2^-}  \frac{ \tr\Big[ \mathbb{E}\left[ \bm{f}(X) ^2\right] - \left( \mathbb{E}\left[\bm{f}(X)^p\right]^{2/p} \right) \Big] }{2-p} = \frac{\Ent(\bm{f}^2)}{2},
	\]
	where the last identity follows from Lemma \ref{lemm:Var_Ent}.
	Similarly, the right-hand side gives
	\[
	\lim_{p\rightarrow 2^-} \frac{ (2-p)\mathcal{E}(\boldsymbol{f})\cdot d^{1-2/p} + \tr \mathbb{E}\left[\boldsymbol{f}^2\right]\cdot(1-d^{1-2/p}) }{2-p} = \mathcal{E}(\bm{f})+ \frac{\log(d)}{2} \cdot \tr \mathbb{E}\left[\bm{f}^2\right]
	\]
	as established.
\end{proof}

\begin{coro}[Matrix Gaussian Logarithmic Sobolev Inequalities]\label{coro:Log_Gau}
	Assume that $X$ is a vector of independent and identical standard Gaussian random variables on $\mathbb{R}^n$ and let $\boldsymbol{f}:\mathbb{R}^n \rightarrow \mathbb{M}_d^+$ be an arbitrary matrix-valued function of $X$. Then,
	\[
	\Ent (\bm{f}^2) \leq  2 \sum_{i=1}^n \mathbb{E} \left[ \big\| \mathsf{D}_{{X}_i}\bm{f}[{X}]\big\|_2^2 \right] + \log(d) \cdot \tr \mathbb{E}\left[\bm{f}^2\right].
	\]
\end{coro}

\begin{remark}\label{remark:LSI}
	Denoted by LS($C,D$) (see e.g.~\cite[Section 5.1]{BGL13}) the set of  log-Sobolev inequalities with constants $C>0$, $D\geq 0$:
	\begin{align*} 
	\Ent(f^2)\leq 2C \mathcal{E}(f) + D \mathbb{E}[f^2]. 
	\end{align*}
	When $D=0$, the log-Sobolev inequality is called \emph{tight}; otherwise, it is called \emph{defective}.
	It is well known that the best constants of the classical log-Sobolev inequalities for symmetric Bernoulli random variables and standard Gaussian random variables are $(C,D)=(1,0)$ \cite{Bon70, Gro75}. However, numerical simulation shows that examples ($d>1$) exist for matrix-valued functions so that:
	$\Ent (\bm{f}^2) >  2 \mathcal{E}(\bm{f})$.
	In Corollary \ref{coro:Log}, we establish the log-Sobolev inequality with constant $(C,D)=(1,\log d)$. 

	We also emphasize that such \emph{defectiveness} in the quantum case has been proved by a recent paper \cite{BR18}.
	Moreover, whether the established constant $(C,D)=(1,\log d)$ is optimal for matrix-valued functions with Bernoulli random variables is still open.
	
	\Endremark
\end{remark}

\section{Entropic Inequality for Classical-Quantum Ensembles} \label{sec:FI}
In this section, we connect the matrix $\Phi$-entropies with quantum information theory and present a functional inequality for the classical-quantum (c-q) ensembles that undergo a special Markov evolution.

We follow the notation in Refs.~\cite{Rag13,Rag14}.
Let $\mathcal{X}$ be a sample space. We denote by $\mathscr{P}(\mathcal{X})$ the set of all probability distributions on $\mathcal{X}$ and by $\mathscr{P}_*(\mathcal{X})$ the subset of $\mathscr{P}(\mathcal{X})$ which consists of all strictly positive distributions. The set of all $d\times d$ matrix-valued functions on $\mathcal{X}$ is denoted by $\mathscr{F}(\mathcal{X})$; $\mathscr{F}_*(\mathcal{X})$ and $\mathscr{F}_*^0(\mathcal{X})$ are the subsets of $\mathscr{F}(\mathcal{X})$ consisting of all strictly positive and non-negative functions, respectively.

Any \emph{classical discrete channel} or \emph{Markov kernel} with input alphabet $\mathnormal{\mathcal{X}}$ and output alphabet $\mathcal{Y}$ can be described by a transition probabilities $\{K(y|x): x \in\mathcal{X},\, y\in\mathcal{Y} \}$.
For any probability distribution $\mu$ defined on the alphabet $\mathcal{X}$, we denote the channel acting on $\mu$ from the right and acting on matrix-valued functions $\bm{f}\in\mathscr{F}(\mathcal{Y})$ respectively by 
\begin{align} \label{eq:kernel}
\mu K(y) \triangleq \sum_{x\in\mathcal{X}} \mu(x) K(y|x), \quad y\in\mathcal{Y}, \quad \text{and}\quad
K \bm{f}(x) \triangleq \sum_{y\in\mathcal{Y}} K(y|x) \bm{f}(y), \quad x\in\mathcal{X}.
\end{align}
The set of all classical channels is denoted by $\mathscr{M}(\mathcal{Y}|\mathcal{X})$.
If $\mu\otimes K \in \mathscr{P}(\mathcal{X}) \times \mathscr{M}(\mathcal{Y}|\mathcal{X})$ denotes the distribution of a random pair $(X,Y) \in \mathcal{X}\times \mathcal{Y}$ with $P_X = \mu$ and $P_{Y|X}=K$, then 
\begin{align} \label{eq:K_exp}
K \bm{f}(x) = \mathbb{E}[ \bm{f}(Y)| X=x ] 
\end{align}
for any $\bm{f}\in\mathscr{F}(\mathcal{Y})$ and $x\in\mathcal{X}$.
We say that a pair $(\mu,K)\in  \mathscr{P}(\mathcal{X}) \times \mathscr{M}(\mathcal{Y}|\mathcal{X})$ is \emph{admissible} if $\mu \in\mathscr{P}_*(\mathcal{X})$ and $\mu K \in \mathscr{P}_*(\mathcal{Y})$.
Hence the \emph{backward} or \emph{adjoint} channel $K^* \in \mathscr{M}(\mathcal{X}|\mathcal{Y})$ can be defined by
\begin{align} \label{eq:adjoint}
K^*(x|y) = \frac{ K(y|x) \mu(x) }{\mu K(y)}, \quad (x,y)\in\mathcal{X}\times \mathcal{Y}.
\end{align}
If $(X,Y) \sim \mu \times K$, it follows that $K^* = P_{X|Y}$ and
\begin{align} \label{eq:K_exp_adjoint}
K^* \bm{f}(y) = \mathbb{E}[ \bm{f}(X)| Y=y ]
\end{align}
for any $\bm{f}\in\mathscr{F}(\mathcal{X})$ and $y\in\mathcal{Y}$.

Define the conditional matrix $\Phi$-entropy of $\bm{Z}$ given ${Y}$ which takes values in any Polish space can be defined by
\begin{align} \label{eq:cond_mEnt}
H_\Phi(\bm{Z}|{Y}) \triangleq \tr \mathbb{E}\left[ \Phi (\bm{Z}) |{Y} \right] - \tr\left[ \Phi\left( \mathbb{E} \left[ \bm{Z} |{Y} \right] \right)\right].
\end{align}
Combining the definition of matrix $\Phi$-entropies with \eqref{eq:cond_mEnt} immediately gives the following law of total variance:
\begin{align} \label{eq:law_total}
\begin{split}
H_\Phi(\bm{Z}) 
&= \mathbb{E}_{Y } \left[ H_\Phi(\bm{Z}|Y ) \right] + H_\Phi\left( \mathbb{E} \left[ \bm{Z} | {Y} \right] \right).
\end{split}
\end{align}

Fix $\Phi(u)=u \log u $ and assume that the distribution $\mu \in \mathscr{P}(\mathcal{X})$ is defined on a discrete space $\mathcal{X}$. If we consider a random matrix $\bm{\rho}_X$ to be an ensemble of classical-quantum (c-q) states $ (\mu,\bm{\nu}) \triangleq \{(\mu(x), \bm{\rho}_x)\}_{x\in\mathcal{X}}$,
where each $\bm{\rho}_x \succeq 0$ and $\Tr \bm{\rho}_x=1$, then its $\Phi$-entropy is related to the \emph{Holevo quantity} of $\{(\mu(x), \bm{\rho}_x)\}_{x\in\mathcal{X}}$:
\begin{align*}
d \cdot H_{u\log u}(\bm{\rho}_X) 
&\equiv \sum_{x\in\mathcal{X}} \mu(x) \Tr \left[ \bm{\rho}_x \log \bm{\rho}_x \right] - \Tr \left[ \bar{\bm{\rho}} \log \bar{\bm{\rho}} \right]
= \sum_{x\in\mathcal{X}} \mu(x) \cdot S\left( \bm{\rho}_x \| \bar{\bm{\rho}} \right)
=: \chi(\mu,\bm{\nu}),
\end{align*}
where $\bar{\bm{\rho}}=\mathbb{E}_\mu[\bm{\rho}_X]=\sum_{x\in\mathcal{X}} \mu(x)\bm{\rho}_x$ and $S(\bm{\rho}\|\bm{\sigma})\triangleq \Tr \bm{\rho} (\log \bm{\rho} - \log \bm{\sigma})$ is the quantum relative entropy.

{Denote by $\bm{\rho}_Y\equiv\{\mu'(y),\bm{\rho}'_y\}_{y\in\mathcal{Y}}$  the resulting random matrix of $\bm{\rho}_X$ that undergoes a Markov evolution $K$}
by the rule:
\begin{align}
\begin{split} \label{eq:K}
\{\mu(x)\}_{x\in\mathcal{X}} &\mapsto \{\mu K (y)\}_{y\in\mathcal{Y}} = \left\{ \sum_{x\in\mathcal{X}} \mu(x) K(y|x) \right\}_{y\in\mathcal{Y}}
=: \{\mu'(y)\}_{y\in\mathcal{Y}} ;\\
\{\bm{\rho}_x\}_{x\in\mathcal{X}} &\mapsto \{ K^* \bm{\nu} (y) \}_{y\in\mathcal{Y}} = \left\{ \sum_{x\in\mathcal{X}}  K^*(x|y) \bm{\rho}_x \right\}_{y\in\mathcal{Y}}
=: \{\bm{\rho}'_y\}_{y\in\mathcal{Y}}.
\end{split}
\end{align}
Note that 
each $\bm{\rho}'_y$ can be interpreted as the conditional expectation $\mathbb{E}_{K^*}[\bm{\rho}_X| Y=y]$, which is a post-selection state with the probability law $\{K^*(x|y)\}_{x\in\mathcal{X}}$. We also have the following relationship between the $\Phi$-entropy of $\bm{\rho}_Y$ and the Holevo quantity of $(\mu',\bm{\nu}') \triangleq \{(\mu' (y),\bm{\rho}'_y \}_{y\in\mathcal{Y}}$:
\begin{align}
d\cdot H_{u\log u}(\bm{\rho}_Y) 
&= \sum_{y\in\mathcal{Y}} \mu'(y) \Tr \left[ \bm{\rho}'_y \log \bm{\rho}'_y \right] - \Tr \left[ \overline{\bm{\rho}'} \log \overline{\bm{\rho}'} \right]
=: \chi(\mu',\bm{\nu}'),
\end{align}
where $\overline{\bm{\rho}'}  =\mathbb{E}_{\mu'}[\bm{\rho}_Y]=\sum_{y\in\mathcal{Y}}\mu'(y)\bm{\rho}'_y$.

Now for any $\mu \in \mathscr{P}(\mathcal{X})$ and $K\in\mathscr{M}(Y|X)$, we define the constant:
\begin{align} \label{eq:eta}
\eta_\Phi(\mu,K) \triangleq 
\sup_{\bm{\nu}: \chi(\mu,\bm{\nu}) \neq 0  }
\frac{\chi(\mu',\bm{\nu}')}{\chi(\mu,\bm{\nu})}.
\end{align}

By Jensen's inequality, it can be shown that $ 0\leq \eta_\Phi(\mu,K) \leq 1$ (see Lemma \ref{lemm:eta}). 
Therefore, we relate $\eta_\Phi(\mu,K)$ to the following functional inequality of the matrix $\Phi$-entropies.

\begin{prop}[Functional Form for C-Q Ensembles] \label{prop:functional_SDPI_cqstate}
	Fix an admissible pair $(\mu,K)$ and let $(X,Y)$ be a random pair with probability law $\mu \otimes K$. Then $\eta_\Phi(\mu,K) \le c$ if and only if the inequality
	\begin{align}\label{eq:entropy_production_inequality_cqstate}
	H_\Phi[\bm{f}(X)] \le \frac{1}{1-c} \mathbb{E}\left( H_\Phi[\bm{f}(X)|Y]\right)
	\end{align}
	holds for all non-constant classical-quantum states $\bm{f}:\mathcal{X}\to \mathcal{Q}(\mathbb{C}^d)$, where we denote by $\mathcal{Q}(\mathbb{C}^d)$ the set of density operators on $\mathbb{C}^d$. 
	In particular, Eq.~\eqref{eq:entropy_production_inequality_cqstate} can be expressed in terms of Holevo quantities:
	\begin{align}
	\chi(\mu, \bm{\nu}) &\le 
	\frac{1}{1-c} \mathbb{E}_{Y} \left[ \chi(K^* , \bm{\nu}|Y)  \right]
	\end{align}	
	where the expectation $\mathbb{E}_Y$ is taken with respect to $\{\mu K(y) \}_{y\in\mathcal{Y}}$.
	
	Moreover,
	\begin{align}\label{eq:functional_SDPI_cqstate}
	\eta_\Phi(\mu,K) 
	&= 1 - \inf\left\{ \frac{\mathbb{E}\left[ H_\Phi(\bm{f}(X)|Y)\right]}{H_\Phi(\bm{f}(X))} : \bm{f} \neq {\rm const}\right\}.
	\end{align}
\end{prop}
\begin{proof}
	The inequality $\chi(\mu' , \bm{\nu}') \le c \cdot \chi(\mu, \bm{\nu})$ is equivalent to 
	\begin{align} \label{eq:Prop1}
	\begin{split}
	H_\Phi( K^* \bm{f}(Y) ) &\leq cH_\Phi ( \bm{f}(X)) \\
	&= c\left( \mathbb{E} \left[ H_\Phi( \bm{f}(X) |Y ) \right] + H_\Phi\left( \mathbb{E} \left[ \bm{f}(X)|Y \right] \right) \right)\\
	&= c\left( \mathbb{E} \left[ H_\Phi( \bm{f}(X) |Y ) \right] + H_\Phi\left( K^* \bm{f}(Y) \right) \right),
	\end{split}
	\end{align}
	where we use the identity of the law of total variance, Eq.~\eqref{eq:law_total}, and the property of the backward channel, Eq.~\eqref{eq:K_exp_adjoint}, from which we obtain
	\begin{align} 
	H_\Phi( K^* \bm{f}(Y) ) \leq \frac{c}{1-c} \mathbb{E} \left[ H_\Phi( \bm{f}(X) |Y ) \right],
	\end{align}
	and hence,
	\begin{align}
	H_\Phi[\bm{f}(X)]  &= \mathbb{E} \left[ H_\Phi( \bm{f}(X) |Y ) \right] + H_\Phi\left( K^* \bm{f}(Y) \right) 
	\leq \frac{1}{1-c} \mathbb{E}\left( H_\Phi[\bm{f}(X)|Y]\right).
	\end{align}
\end{proof}

Raginsky showed, in recent work \cite{Rag13,Rag14},  that if $f=\mathrm{d} \nu / \mathrm{d} \mu$ (i.e.~a Radon-Nikodym derivative), then
$K^ * f = \frac{ \mathrm{d} (\nu K) }{ \mathrm{d} (\mu K) }$.
Moreover, the constant $\eta_\Phi(\mu,K)$ in Eq.~\eqref{eq:eta} corresponds to the (classical) \emph{strong data processing inequality} (SDPI):
\begin{align} \label{eq:eta2}
\eta_\Phi(\mu,K) \triangleq \sup_{{\nu}\neq \mu} \frac{D(\nu K , \mu K )}{D(\nu,\mu)},
\end{align}
where $D(\nu,\mu)$ is the classical Kullback-Leibler divergence of $\mu$ and $\nu$. 
We remark that Proposition~\ref{prop:functional_SDPI_cqstate} can be viewed as a generalization of  Raginshky's SDPI result to case of quantum ensembles undergoing the Markov evolution $K$ described in Eq.~\eqref{eq:K}.

\section{Conclusions} \label{sec:conclusion}

In this paper, we extend the work of Gross \cite{Gro75} to study the $\Phi$-Sobolev and Poincar\'e inequalities on matrix-valued functions. The uncertainty measure adopted in this work is the matrix $\Phi$-entropy \cite{CT14}, which interpolates between the variance and conventional entropies of random matrices. The Dirichlet energy is defined as the average of local influences on every bit position. We note, that if the considered domain is a Boolean hypercube (i.e.~$\mathcal{X}=\{0,1\}^n$), then the our notion coincides with the Dirichlet form induced by the Markov jumping process \cite{CHT15}.
Our results generalise the classical functional inequalities to the matrix  Poincar\'e inequality for separable convex functions and Gaussian unitary ensembles. Moreover, we establish the matrix $\Phi$-Sobolev inequalities for symmetric Bernoulli and standard Gaussian distributions. Unlike its classical counterpart \cite{LO00}, the derived inequality is defective, which shows the matrix version is more involved, but it reduces to the classical case when $d=1$.

The random matrix framework has immediate applications in quantum information theory and mathematics.
For example, by relating the matrix $\Phi$-entropy to the celebrated Holevo quantity, we obtain a formula for the strong data processing inequality. In a follow-up work \cite{CHT15}, we are able to upper bound the convergence rate of the dynamical evolutions of a quantum ensemble by employing the matrix Efron-Stein inequality and a modified log-Sobolev inequality.
Furthermore, the studied model naturally occurs in the probabilistic context of matrix-valued stochastic processes \cite{Dys62b,CFM+11}. Solving the matrix functional inequalities of the Markovian matrix-valued processes provides a way to characterize their long-term behaviours. 


\appendix

\section{Miscellaneous Lemmas}\label{app_lemmas}

\begin{lemm}
	\label{lemm:plus}
	Let $\bm{X}$ be a random matrix taking values in $\mathbb{M}^\text{sa}$, and let $\bm{Y}$ be independently and identically distributed as $\bm{X}$. 
	Then for each natural number $q\geq 1$, 
	\begin{align} \label{eq:plus1}
	\mathbb{E}\left[ \left| \bm{X} - \mathbb{E}\bm{X} \right|^q \right]
	= \mathbb{E}\left[ \left( \bm{X} - \mathbb{E}\bm{X} \right)_+^q \right]
	+ \mathbb{E}\left[ \left( \bm{X} - \mathbb{E}\bm{X} \right)_-^q \right]
	\end{align}
	and
	\begin{align} \label{eq:plus2}
	\frac12 \, \mathbb{E} \left[ \left| \bm{X}-\bm{Y} \right|^q\right]
	= \mathbb{E} \left[ \left( \bm{X}-\bm{Y} \right)_+^q\right]
	= \mathbb{E} \left[ \left( \bm{X}-\bm{Y} \right)_-^q\right]
	\end{align}
	In particular,
	\begin{align} \label{eq:plus2b}
	\mathbb{E}\left[ \left( \bm{X} - \mathbb{E}\bm{X} \right)^2 \right]
	= \frac12 \, \mathbb{E} \left[ \left( \bm{X}-\bm{Y} \right)^2\right].
	\end{align}
\end{lemm}
\begin{proof}
	For each realisation $X$ of $\bm{X}$ in $\mathbb{M}^\text{sa}$, $X=X_+-X_-$ for some $X_+, X_-\succeq 0$ and ${X}_+  {X}_- = \bm{0}.$ Slightly abusing the notation, we hence use $\bm{X}_+$ and $\bm{X}_-$ to denote the positive and negative decomposition of their realisations of $\bm{X}$.

	Therefore, for each natural number $q\geq 1$, 
	\begin{align*}
	\mathbb{E}\left[ \left| \bm{X} - \mathbb{E}\bm{X} \right|^q \right]
	&= \mathbb{E} \left[ \left( \left( \bm{X} - \mathbb{E}\bm{X} \right)_+ + \left( \bm{X} - \mathbb{E}\bm{X} \right)_- \right)^q \right] 
	= \mathbb{E}\left[ \left( \bm{X} - \mathbb{E}\bm{X} \right)_+^q \right]
	+ \mathbb{E}\left[ \left( \bm{X} - \mathbb{E}\bm{X} \right)_-^q \right].	
	\end{align*}
	Likewise, we have
	\begin{align*}
	\frac12 \,  \mathbb{E} \left[ \left| \bm{X}-\bm{Y} \right|^q\right]
	&= \frac12 \, \mathbb{E} \left[ \left( \left( \bm{X}-\bm{Y} \right)_+ + \left( \bm{Y}-\bm{X} \right)_+ \right)^q\right] 
	=  \mathbb{E} \left[ \left( \bm{X}-\bm{Y} \right)_+^q\right].
	\end{align*}
	The last line follows since $\bm{Y}$ is an identical copy of $\bm{X}$.
	
	Following the same reasoning, we have $|\bm{X}|=\bm{X}_- + (-\bm{X})_-$, and thus $	\frac12 \mathbb{E} \left[ \left| \bm{X}-\bm{Y} \right|^q\right]
	= \mathbb{E} \left[ \left( \bm{X}-\bm{Y} \right)_-^q\right]$.
	Finally, Eq.~(\ref{eq:plus2b}) follows from elementary calculations:
	\begin{align*}
	\frac12 \mathbb{E} \left[ \left( \bm{X}-\bm{Y} \right)^2\right]
	&= \frac12 \mathbb{E} \left[ \bm{X}^2-\bm{X}\bm{Y}-\bm{Y}\bm{X} + \bm{Y}^2 \right]
	= 	\mathbb{E}\left[ \left( \bm{X} - \mathbb{E}\bm{X} \right)^2 \right].
	\end{align*}
\end{proof}

\begin{lemm}
	[Central Limit Theorem of Gaussian Unitary Ensembles] \label{lemm:GUE}
	Let $\{\epsilon_j\}_j$ be a series of Rademacher variables, and let $\{\bm{W}_j\}_j$, $\{\bm{W}_j ' \}_j$ be $d\times d$ matrices whose entries are sampled independently from the Rademacher variables.
	Let
	\[
	\bm{Y}_j =  \frac{\ \left(\bm{W}_j+\mathrm{i}\cdot \bm{W}_j'\right) + \left(\bm{W}_j+\mathrm{i}\cdot \bm{W}_j'\right)^\dagger }{2},
	\]
	and
	$\bm{S}_m \triangleq \frac{1}{\sqrt{m}} \sum_{j=1}^m \epsilon_j \bm{Y}_j, $
	where $\{\epsilon_j\}_j$ are Rademacher variables again. 
	If $m$ tends to infinity, then $\bm{S}_m$ converges in distribution to a $d\times d$ matrix in the Gaussian unitary ensemble.
\end{lemm}
\begin{proof}
	It is clear from the central limit theorem that the diagonal entries converge to a standard real Gaussian variable, while the upper-triangular entries converge to a complex Gaussian variables with zero mean and unit variance.
	Next, we show that the correlation between any (non-identical) entry vanishes as $m$ goes to infinity.
	That is, for every $(k,l)\neq (k',l')$
	\begin{align*}
	\mathbb{E}_{\epsilon_1,\ldots,\epsilon_m} \left[ S_{m}^{(kl)} S_{m}^{(k'l')} \right]
	= \frac1m \sum_{j=1}^m  Y_{j}^{(kl)} Y_j^{(k'l')},
	\end{align*}
	from which we apply the strong law of large numbers to obtain
	\begin{align*}
	\lim_{m\to \infty} \frac1m \sum_{j=1}^m  Y_{j}^{(kl)} Y_j^{(k'l')} 
	= \mathbb{E} \left[ Y \cdot Y' \right] 
	= \mathbb{E} \left[ Y \right] \cdot \mathbb{E} \left[ Y' \right]
	=0 \quad \text{almost surely,}
	\end{align*}
	where we denote by $Y$ (resp.~$Y'$) the random variable that the sequences $\{Y_j^{(kl)}\}_j$ (resp.~$\{Y_j^{(k'l')}\}_j$) are sampled from. It is easy to see that $Y$ and $Y'$ are independent zero-mean random variables. 
	Therefore, the entries are mutually independent and $\lim_{m\rightarrow \infty} \bm{S}_m $ belongs to the Gaussian unitary ensemble.
\end{proof}

\begin{lemm}
	[Parseval's Identity for Matrix-Valued Functions] \label{lemm:Parseval}
	For every matrix-valued function $\boldsymbol{f}:\{0,1\}^n \rightarrow \mathbb{M}_d$,
	we have the following identity
	\[
	\mathbb{E} \left[ \boldsymbol{f} (X)^2 \right] = \sum_{ S\subseteq [n] } \widehat{\boldsymbol{f}} (S)^2,
	\]
	where the expectation is taken uniformly over all $X\in\{0,1\}^n$.
\end{lemm}
\begin{proof}
	With the Fourier expansion of the matrix-valued function $\boldsymbol{f}$, it follows that
	\begin{align*}
	\mathbb{E} \left[ \boldsymbol{f} (X)^2 \right]
	= \mathbb{E} \left[ \boldsymbol{f} (X) \cdot \left( \sum_{ S\subseteq [n] } \widehat{\boldsymbol{f}}(S) \chi_S(X) \right) \right] 
	= \sum_{ S\subseteq [n] } \widehat{\boldsymbol{f}}(S) \cdot \mathbb{E} \left[ \boldsymbol{f}(X)\chi_S (X) \right] 
	= \sum_{ S\subseteq [n] } \widehat{\boldsymbol{f}} (S)^2.
	\end{align*}
\end{proof}

\begin{lemm}
	\label{lemm:Var}
	With the prevailing assumptions, and every $\boldsymbol{f}:\{0,1\}^n \rightarrow \mathbb{M}_d^\text{sa}$, we have
	\[
	\sum_{ S\subseteq [n] } \tr \left[ |S| \widehat{\boldsymbol{f}}(S)^2 \right] = \mathcal{E}(\boldsymbol{f}).
	\]
\end{lemm}
\begin{proof}
	For every $n$-tuple $x\triangleq (x_1,\ldots,x_n)\in\{0,1\}^n$, denote $\overline{x}^{(i)}\triangleq (x_1,\ldots, x_{i-1}, 1-x_i,x_{i+1},\ldots,x_n )$.
	For every $i\in[n]$, introduce the matrix-valued function
	$\boldsymbol{g}_i(x)= \frac{ \boldsymbol{f}(x)-\boldsymbol{f}(\overline{x}^{(i)})}2$.
	Then, for every $S\subseteq [n]$, it can be observed that
	\begin{align*}
	&\widehat{\boldsymbol{g}}_i(S) = \mathbb{E} \left[ \boldsymbol{g}_i(X) \chi_S(X) \right]
	= \frac12 \mathbb{E} \left[ \left( \boldsymbol{f}(X)-\boldsymbol{f}(\overline{X}^{(i)}) \right) \cdot (-1)^{\sum_{j \in S} X_j} \right]
	= 
	\begin{cases}
	0 \quad  &\text{if } i \notin S\\
	\widehat{\boldsymbol{f}}(S) \quad &\text{if } i\in S.
	\end{cases}
	\end{align*}
	Apply Parseval's identity, Lemma \ref{lemm:Parseval} to obtain
	\[
	\mathbb{E}\left[ \boldsymbol{g}_i(X)^2 \right] = \sum_{S\subseteq [n]} \widehat{\boldsymbol{g}}_i (S)^2
	= \sum_{S:i\in S} \widehat{\boldsymbol{f}}(S)^2.
	\]
	Finally, since $X$ is uniformly distributed, $\mathcal{E}(\boldsymbol{f})$ can be rewritten as
	\begin{align*}
	\mathcal{E}(\boldsymbol{f}) 
	&= \frac12 \tr\mathbb{E}\left[\sum_{i=1}^n \left( \boldsymbol{f}({X})-\boldsymbol{f}\left(\widetilde{{X}}^{(i)}\right)\right)^2\right]
	=\frac14 \tr\mathbb{E}\left[\sum_{i=1}^n \left( \boldsymbol{f}({X})-\boldsymbol{f}\left(\overline{{X}}^{(i)}\right)\right)^2\right]\\
	&= \sum_{i=1}^n \tr \mathbb{E}\left[ \boldsymbol{g}_i(X)^2 \right]
	= \sum_{i=1}^n \sum_{S:i\in S} \tr \left[ \widehat{\boldsymbol{f}}(S)^2 \right] 
	=\sum_{ S\subseteq [n] } \tr \left[ |S| \widehat{\boldsymbol{f}}(S)^2 \right]. 
	\end{align*}
	This completes the proof.
\end{proof}

\begin{lemm} \label{lemm:Var_Ent}
	Let $\bm{Z}$ be a random matrix taking values in $\mathbb{M}^+$ such that $\|\bm{Z}\|_\infty<\infty$. For $p\in [1,2)$, we define the matrix-valued $p$-variance of $\bm{Z}$ by
	$\textnormal{\textbf{Var}}_p[\bm{Z}] \triangleq \mathbb{E}\big[\bm{Z}^2\big] - \Big( \mathbb{E}\big[\bm{Z}^p\big] \Big)^{2/p}$.
	It follows that
	\[
	\lim_{p\to 2^-} \frac{\textnormal{\textbf{Var}}_p[\bm{Z}]}{2-p}
	= \frac12 \mathbb{E}\Big[ \bm{Z}^2\log \bm{Z}^2 \Big] - \frac12 \mathbb{E}\big[\bm{Z}^2\big] \cdot \log \mathbb{E}\big[ \bm{Z}^2 \big].
	\]
\end{lemm}
\begin{proof}
	We first prove a formula for the matrix differentiation.
	Denote by $\bm{A}=\bm{A}(p)$ a Hermitian matrix which depends on the real parameter $p$.
	Then we aim to solve the derivative of $\bm{A}^{2/p}$ with respect to $p$.
	Let $\bm{Y} = \bm{A}^{2/p}$. Then $\log \bm{Y} =  \log \bm{A} \cdot 2/p$.
	Differentiating on both sides with respect to $p$ and applying the chain rule of the Fr\'echet derivatives, the above expression leads to
	\begin{align} \label{eq:Var_Ent1}
	\begin{split}
	\frac{\mathrm{d}}{\mathrm{d}p} \log \bm{Y}
	&= \int_0^\infty (s\bm{I} + \bm{Y})^{-1} \cdot \frac{\mathrm{d}}{\mathrm{d}p} \bm{Y} \cdot (s\bm{I} + \bm{Y})^{-1} \mathrm{d} s 
	=\frac{\mathrm{d}}{\mathrm{d}p}  \log \bm{A}\cdot 2/p \\
	&= -\frac2{p^2} \log\bm{A}
	+ \frac2p \int_0^\infty (t\bm{I} + \bm{A})^{-1} \cdot \frac{\mathrm{d}}{\mathrm{d}p} \bm{A} \cdot (t\bm{I} + \bm{A})^{-1} \mathrm{d} t.
	\end{split}
	\end{align}
	Note that $\mathsf{T}_{\bm{D}}(\bm{K})\triangleq \int_0^\infty (s\bm{I} + \bm{D})^{-1} \bm{K} (s\bm{I} + \bm{D})^{-1} \mathrm{d} s$ is called the \emph{Bogoliubov-Kubo-Mori operator} and its inverse is well-known to be
$	\mathsf{T}_{\bm{D}}^{-1}(\bm{L}) = \int_0^1 \bm{D}^s \bm{L} \bm{D}^{1-s} \mathrm{d} s$ (see e.g. \cite[Appendix C.2]{HKP+13}),
	from which Eq.~\eqref{eq:Var_Ent1} yields
	\begin{align*}
	\frac{\mathrm{d}}{\mathrm{d}p} \bm{Y}
	&= \int_0^1 \bm{Y}^s \left[   -\frac2{p^2} \log\bm{A}
	+ \frac2p \int_0^\infty (t\bm{I} + \bm{A})^{-1} \cdot \frac{\mathrm{d}}{\mathrm{d}p} \bm{A} \cdot (t\bm{I} + \bm{A})^{-1} \mathrm{d} t   \right] \bm{Y}^{1-s} \,\mathrm{d} s \\
	&= -\frac2{p^2} \bm{A}^{2/p} \cdot \log\bm{A}
	+  \frac2p \int_0^1 \int_0^\infty \bm{A}^{2s/p} 
	(t\bm{I} + \bm{A})^{-1} \cdot \frac{\mathrm{d}}{\mathrm{d}p} \bm{A} \cdot (t\bm{I} + \bm{A})^{-1}     \bm{A}^{2/p-2s/p} \, \mathrm{d} t \,\mathrm{d} s.
	\end{align*}
	Now, by taking $\bm{A} \equiv \mathbb{E}\big[ \bm{Z}^p \big]$ we have
	$\frac{\mathrm{d}}{\mathrm{d}p} \bm{A} = \frac{\mathrm{d}}{\mathrm{d}p}  \mathbb{E}\big[\bm{Z}^p\big] = \mathbb{E}\big[ \bm{Z}^p \cdot \log \bm{Z} \big]$
	and
	\begin{align} \label{eq:Var_Ent2}
	\begin{split}
	\frac{\mathrm{d}}{\mathrm{d}p} \Big( \mathbb{E}\big[\bm{Z}^p\big] \Big)^{2/p}
	&= -\frac2{p^2} \Big( \mathbb{E}\big[\bm{Z}^p\big] \Big)^{2/p} \log \mathbb{E} \big[ \bm{Z}^p \big] \\
	&+ \frac2p \int_0^1 \int_0^\infty \bm{A}^{2s/p} 
	(t\bm{I} + \bm{A})^{-1} \cdot \mathbb{E}\big[ \bm{Z}^p \cdot \log \bm{Z} \big] \cdot (t\bm{I} + \bm{A})^{-1}     \bm{A}^{2/p-2s/p} \, \mathrm{d} t \,\mathrm{d} s.
	\end{split}
	\end{align}
	Finally, we are ready to prove our claim. L'H\^opital's rule implies
	\begin{align*}
	&\lim_{p\to 2^-} \frac{\textnormal{\textbf{Var}}_p[\bm{Z}]}{2-p}
	= \left. 	\frac{\mathrm{d}}{\mathrm{d}p} \Big( \mathbb{E}\big[\bm{Z}^p\big] \Big)^{2/p}\right|_{p= 2} \\
	&= - \frac12 \mathbb{E}\big[\bm{Z}^2\big] \cdot \log \mathbb{E}\big[ \bm{Z}^2 \big]
	+ \int_0^1 \int_0^\infty \bm{A}^{s} 
	(t\bm{I} + \bm{A})^{-1} \cdot \mathbb{E}\big[ \bm{Z}^p \cdot \log \bm{Z} \big] \cdot (t\bm{I} + \bm{A})^{-1}     \bm{A}^{1-s} \, \mathrm{d} t \,\mathrm{d} s \\
	&= \frac12 \mathbb{E}\Big[ \bm{Z}^2\log \bm{Z}^2 \Big] - \frac12 \mathbb{E}\big[\bm{Z}^2\big] \cdot \log \mathbb{E}\big[ \bm{Z}^2 \big],
	\end{align*}
	completing the proof.
\end{proof}

\begin{lemm} \label{lemm:eta}
	Fix sample spaces $\mathcal{X}$ and $\mathcal{Y}$.
	For every distribution $\mu\in\mathscr{P}(\mathcal{X})$, Markov kernel $K\in\mathscr{M}(\mathcal{Y}|\mathcal{X})$ and matrix-valued function $\bm{f}:\mathcal{X} \to \mathbb{M}_d^+$, we have the following  inequality:
	\begin{align*}
	H_\Phi ( K^* \bm{f}) &= \tr\Big[ \mathbb{E}_{\mu K} \big[ \Phi ( K^* \bm{f}) \big]- \Phi\left( \mathbb{E}_{\mu K} \big[ K^* \bm{f} \big] \right)  \Big]
	\leq \tr\Big[ \mathbb{E}_\mu  \Phi (\bm{f})  - \Phi\left( \mathbb{E}_\mu  \bm{f} \right) \Big]
		= 	H_\Phi ( \bm{f}),
	\end{align*}
	where $\mu K$ and $K^*$ are defined in Eqs.~\eqref{eq:kernel} and \eqref{eq:adjoint}. 
\end{lemm}
\begin{proof}
	We first observe that Jensen's inequality {\cite[Section 5]{FZ07}} holds for all convex functions $\Phi$:
	$\tr \left[ K^* \Phi (\bm{f}) \right] \geq \tr \left[ \Phi \left(K^* \bm{f}\right) \right]$.
	After taking expectation with respect to $\mu K$, direct calculation shows that
	\begin{align*}
	\tr \mathbb{E}_{\mu K}\left[ K^* \Phi\left( \bm{f} \right) \right]
	&= \sum_{y\in\mathcal{Y}} \mu K(y) \cdot \tr \left[ K^* \Phi \circ \bm{f} (y) \right] 
	= \sum_{y\in\mathcal{Y}} \mu K(y) \cdot \tr \left[ \sum_{x\in\mathcal{X}} \frac{K(y|x)\mu(x)}{\mu K (y)} \Phi \big(\bm{f} (x) \big) \right] \\
	&= \sum_{x\in\mathcal{X}} \mu(x)  \tr \left[ \Phi \big(\bm{f} (x) \big) \right] 
	= \tr \mathbb{E}_\mu \left[\Phi( \bm{f}) \right] 
	\geq \sum_{y\in\mathcal{Y}} \mu K (y) \cdot \tr \left[ \Phi\big( K^* \bm{f} \big) \right]\\
	&= \tr \mathbb{E}_{\mu K} \big[ \Phi ( K^* \bm{f}) \big]
	\end{align*}
	(note that we freely interchange the order of trace and expectation by Fubini's theorem).
	Together with the fact that $\mathbb{E}_{\mu K} \big[ K^* \bm{f} \big] = \mathbb{E}_\mu  \bm{f}$, we complete our claim.	
\end{proof}

\section*{Acknowledgements}
MH would like to thank Matthias Christandl, Michael Kastoryano, Robert Koenig, Joel Tropp, and Andreas Winter for their useful comments. 
MH was supported by ARC Future Fellowship under Grant FT140100574 and by US Army Research Office for Basic Scientific Research Grant W911NF-17-1-0401. 
HC was supported by Ministry of Science and Technology Overseas Project for Post Graduate Research under Grant 108-2917-I-564-042.




\end{document}